\documentclass[letterpaper,twocolumn,10pt]{article}
\usepackage{usenix2019,epsfig,endnotes}

\usepackage[title]{appendix}
\usepackage{graphicx}
\usepackage[misc]{ifsym}
\usepackage{cite}
\usepackage{makecell}
\usepackage{balance} 
\usepackage{times}
\usepackage{graphicx}
\usepackage{subfigure}
\usepackage[hyphens]{url}
\usepackage{algorithm}
\usepackage{algorithmicx}
\usepackage{algpseudocode}
\usepackage{amsmath}
\usepackage{amsthm}
\usepackage{amssymb}
\usepackage{mathrsfs}
\usepackage{multirow}
\usepackage{xcolor,colortbl}
\usepackage[normalem]{ulem}

\usepackage{float}
\usepackage{comment}
\usepackage{microtype}
\usepackage{listings}

\DeclareMathAlphabet{\mathcal}{OMS}{cmsy}{m}{n}

\newtheorem{theorem}{Theorem}

\newtheorem{lemma}{Lemma}

\newcommand{\PS}{S}

\newcommand{\NT}{\widetilde{T}}
\newcommand{\NS}{\widetilde{S}}

\newcommand{\Pg}{\mathbf{g}}
\newcommand{\Ng}{\mathbf{\widetilde{g}}}
\newcommand{\SumT}{\sum_{t=1}^{T}}
\newcommand{\AvgT}{\frac{1}{T}\sum_{t=1}^{T}}

\renewcommand{\paragraph}[1]{{\smallskip\noindent\bf #1}}

\begin{document}

\date{}

\title{\Large \bf On the Performance and Convergence of Distributed Stream Processing via Approximate Fault Tolerance
\thanks{An earlier version of the paper appeared in \cite{Huang2016}. In this
extended version, we formally prove that AF-Stream preserves the convergence
guarantees of online learning in distributed stream processing, and also validate our convergence analysis through prototype experiments for different datasets, parameters, and consistency models.} 
}

\author{
{\rm Zhinan Cheng$^1$, Qun Huang$^2$, and Patrick P. C. Lee$^1$} \\
$^1$Department of Computer Science and Engineering, The Chinese University of Hong Kong \\
$^2$State Key Lab of Computer Architecture, Institute of Computing Technology,
Chinese Academy of Science
} 

\maketitle

\thispagestyle{empty}

\begin{abstract}
Fault tolerance is critical for distributed stream processing systems, 
yet achieving error-free fault tolerance often incurs substantial performance
overhead.  We present {\em AF-Stream}, a distributed stream processing system
that addresses the trade-off between performance and accuracy in fault
tolerance.  AF-Stream builds on a notion called {\em approximate fault
tolerance}, whose idea is to mitigate backup overhead by adaptively issuing
backups, while ensuring that the errors upon failures are bounded with
theoretical guarantees.  Specifically, AF-Stream allows users to specify
bounds on both the state divergence and the loss of non-backup streaming
items. It issues state and item backups only when the bounds are reached.
Our AF-Stream design provides an extensible programming model for
incorporating general streaming algorithms as well as exports only few
threshold parameters for configuring approximation fault tolerance.
Furthermore, we formally prove that AF-Stream preserves high
algorithm-specific accuracy of streaming algorithms, and in particular the
convergence guarantees of online learning.  Experiments show that AF-Stream
maintains high performance (compared to no fault tolerance) and high accuracy
after multiple failures (compared to no failures) under various streaming
algorithms. 


\end{abstract}

\section{Introduction}
\label{sec:introduction}

Stream processing becomes an important paradigm for processing data at high
speed and large scale.  As opposed to traditional batch processing that is
designed for static data, stream processing treats data as a continuous
stream, and processes every item in the stream in real-time.  For scalability,
we can make stream processing {\em distributed}, by processing streaming items
in parallel through multiple processes (or workers) or threads.  

Given that failures can happen at any time and at any worker in a distributed
environment, fault tolerance is a critical requirement in distributed stream
processing.  In particular, streaming algorithms often keep internal states in
main memory, which is fast but unreliable.  Also, streaming items are generated
continuously in real-time, and will become unavailable after being processed. 
Thus, we need to provide fault tolerance guarantees for both internal states
and streaming items.  Most existing stream processing systems (e.g., Spark
Streaming \cite{Zaharia2013} and Apache Flink \cite{Carbone2015}) support
{\em error-free} fault tolerance and exactly-once message delivery semantics.
They issue regular backups for both internal states and streaming items, so
that when failures happen, they can recover from the most recent backup and
resume processing as normal.  However, frequent backups
incur significant network and disk I/Os, thereby degrading stream
processing performance.  Some stream processing systems (e.g., Storm
\cite{Toshniwal2014} and Heron \cite{Kulkarni2015}) support {\em best-effort}
fault tolerance and at-most-once message delivery semantics to trade recovery
accuracy for performance, but the drawback is that they can incur unbounded
errors that compromise the correctness of outputs when recovering from
failures. 

We propose {\em AF-Stream}, a distributed stream processing system that 
addresses the trade-off between performance and accuracy in fault tolerance.
AF-Stream builds on a notion called {\em approximate fault tolerance}, whose
idea is to adaptively issue backup operations for both internal states and
unprocessed items, while incurring only {\em bounded} errors after failures
are recovered.  Specifically, AF-Stream estimates the errors upon failures
with the aid of extensible programming interfaces, and issues a backup
operation only when the errors go beyond the user-specified acceptable level.

We justify the trade-off with two observations.
First, to mitigate computational complexities, streaming algorithms tend to
produce ``quick-and-dirty" results rather than exact ones (e.g., data
synopsis) (\S\ref{subsec:classes}).  It is thus feasible to incur small
additional errors due to approximate fault tolerance, provided that the errors
are bounded.  Such errors can often be amortized after the processing of
large-volume and high-speed data streams.
Second, although failures are prevalent in distributed systems, their
occurrences remain relatively infrequent over the lifetime of stream
processing.  Thus, incurring errors upon failures should bring limited
disturbance.  We point out that the power of approximation has been
extensively addressed in distributed computing (e.g., \cite{Agarwal2013,
	Li2014, Rabkin2014, Agarwal2015, Pundir2015, Wei2015, Xing2015}). To 
our knowledge, AF-Stream is the first work that leverages approximation to
achieve fault tolerance in distributed stream processing with bounded error
guarantees.  Our contributions are summarized as follows. 

First, AF-Stream provides an extensible programming model for general
streaming algorithms.  In particular, it exports built-in interfaces and
primitives that make fault tolerance intrinsically supported and transparent
to programmers.  

Second, AF-Stream realizes approximate fault tolerance, which bounds errors
upon failures in two aspects: state divergence and number of lost items.  We
prove that the errors are bounded independent of the number of failures and
the number of workers in a distributed environment.  Also, the error bounds
are tunable with only {\em three} user-configurable threshold parameters to
trade between performance and accuracy.  We give examples on how such
parameters correspond to the error bounds of streaming algorithms.  Note that
AF-Stream adds no error to the common case when no failure happens. 

Third, we analyze how AF-Stream affects the model convergence of online
learning under approximate fault tolerance.  We formally prove that AF-Stream
preserves the iterative-convergent and error-tolerant nature of online
learning \cite{Bottou2007,Dai2015} after recovering from multiple failures.
Our analysis focuses on {\em stochastic gradient descent}
\cite{Bottou1998}, a core primitive for many popular online learning
algorithms \cite{Bottou2010}.  

Fourth, we implement an AF-Stream prototype, with emphasis on optimizing its
inter-thread and inter-worker communications.  Our AF-Stream prototype also
supports different consistency models. 

Finally, we evaluate the performance and accuracy of our AF-Stream prototype
on Amazon EC2 and a local cluster for various streaming algorithms.  AF-stream
only degrades the throughput by up to 4.7\%, 5.2\%, and 0.3\% in heavy hitter
detection, online join, and online logistic regression, respectively, when
compared to disabling backups; meanwhile, its accuracy after 10 failures only
drops by a small percentage (based on algorithm-specific metrics) when
compared to without failures.  Also, for online logistic regression,
AF-Stream only slightly delays model convergence (e.g., by 13.3\% and 20\% for
two real-world traces under eight failures).

The source code of our AF-Stream prototype is available at: 
{\bf http://adslab.cse.cuhk.edu.hk/software/afstream}. 

\section{Streaming Algorithms}
\label{sec:algs}

We first provide an overview of general streaming algorithms.  We then focus
on three classes of streaming algorithms, and identify their commonalities to
guide our system design.

\subsection{Overview and Motivation}
\label{subsec:abstraction}

A streaming algorithm comprises a set of {\em operators} for processing a data
stream of {\em items}.  Each operator continuously receives an input item,
processes the item, and produces one or multiple output items.  It also keeps
an in-memory {\em state}, which holds a collection of values corresponding to
the processing of all received input items.  The state is updated after each
input item is processed.  Also, an operator produces new items with respect to
the updated state.  

If failures never happen, we refer to the produced state and output as the
{\em ideal state} and {\em ideal output}, respectively.  However, the actual
deployment environment is failure-prone.  In this case, an operator needs to
generate its {\em actual state} and {\em actual output} in a fault-tolerant
manner.  Specifically, an operator needs to handle: (i) the missing in-memory
state and (ii) missing unprocessed items.  

Suppose that we want to achieve error-free fault tolerance for an operator.
That is, we ensure that the actual state and output are identical to the ideal
state and output, respectively, after the recovery of a failure. 
This necessitates the availability of a prior state and any following items,
so that when a failure happens, we can retrieve the prior state and resume the
processing of the following items.  In particular, in stream processing, the
items to be processed are generated continuously, so we need to issue periodic
state backups and a backup for {\em every} item in order to achieve error-free
fault tolerance.  For example, Spark Streaming saves each state as a Resilient
Distributed Dataset (RDD) \cite{Zaharia2013} in a mini-batch fashion and saves
each item via write-ahead logging \cite{SparkZeroLoss}.

However, making regular backups for both the state and each item incurs
excessive I/O overhead to normal processing.  
We motivate this claim by evaluating the backup overhead of Spark
Streaming (v1.4.1) versus the mini-batch interval (i.e., the duration within
which items are batched).  Note that the performance of Spark Streaming is
sensitive to the number of items in a mini-batch: having small mini-batches
aggravates backup overhead, while having very large mini-batches increases the
processing time to even exceed the mini-batch interval and makes the system
{\em unstable} (i.e., the system throughput varies significantly)
\cite{Das2014}.  
Thus, for each given mini-batch interval, we tune the stream input rate that
gives the maximum stable throughput.  Figure~\ref{fig:spark} shows the
throughput of Spark Streaming for Grep and WordCount, measured on Amazon EC2
(see \S\ref{sec:eval} for the details on the datasets and experimental setup).
We see that the throughput drops significantly due to backups, for example, by
nearly 50\% for WordCount when the mini-batch interval is 1s.

\begin{figure}[!t]
\centering
\begin{tabular}{c@{\ }c}
	\hspace{-0.1in}
	\includegraphics[width=1.65in]{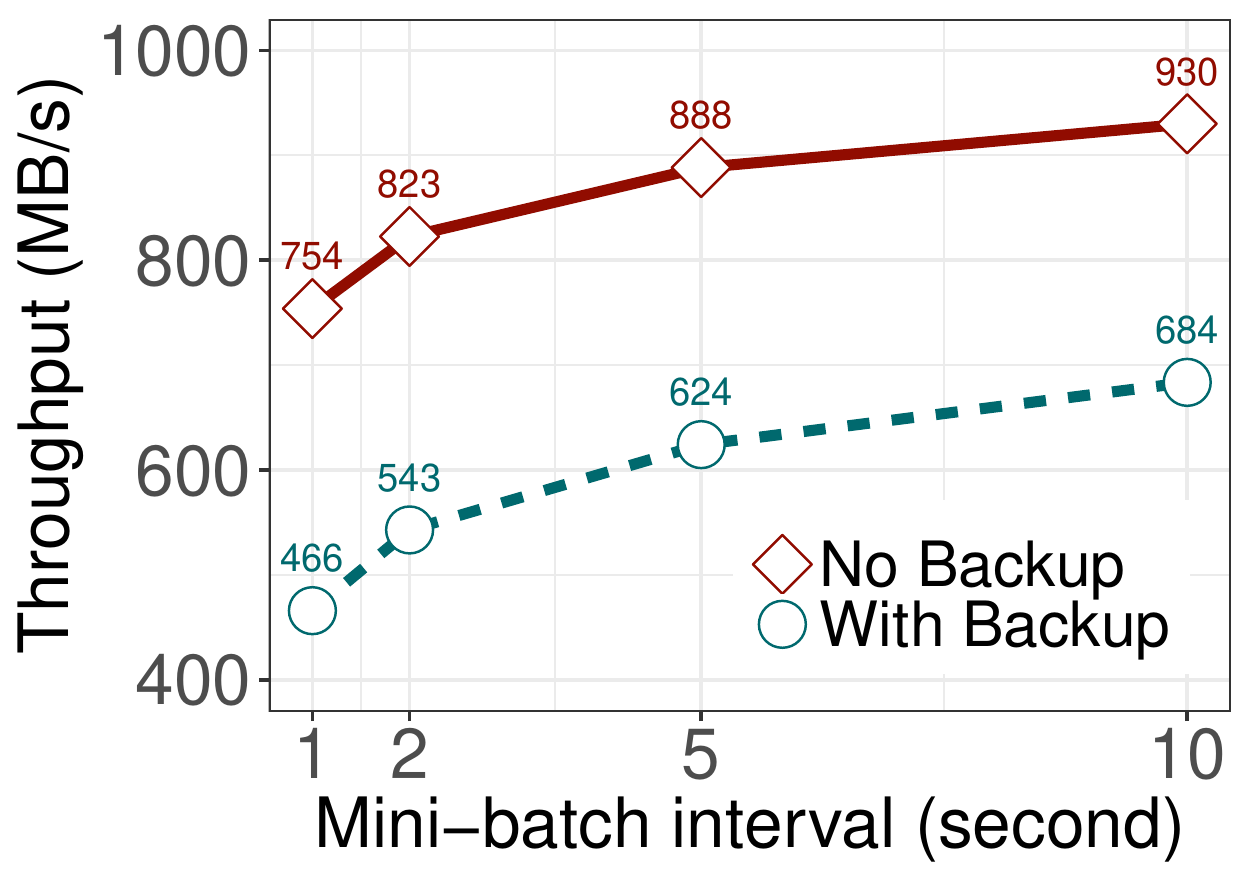} &
	\includegraphics[width=1.65in]{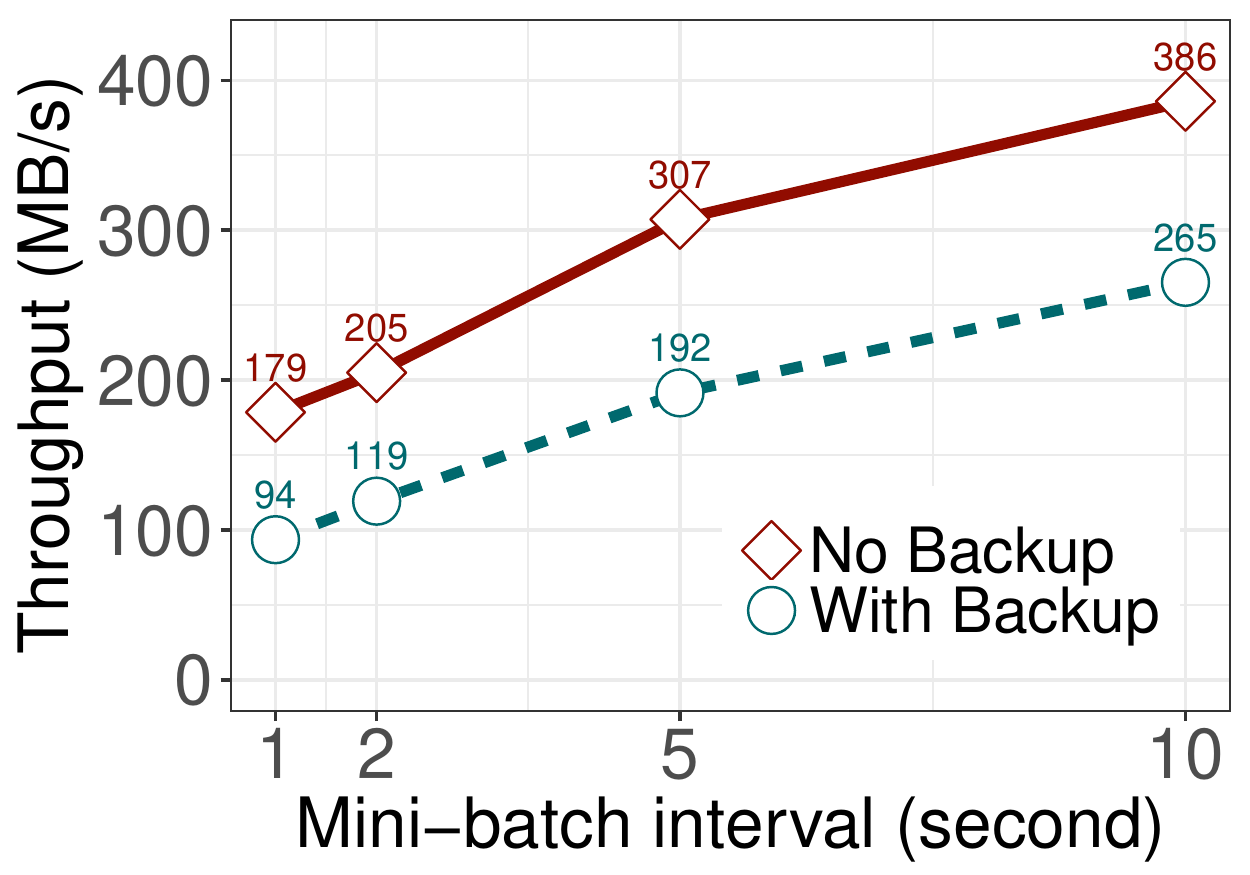}\\
	{\small (a) Grep} &
	{\small (b) WordCount} 
\end{tabular}
\vspace{-6pt}
\caption{Throughput of Spark Streaming with and without backups on Amazon EC2.}
\label{fig:spark}
\end{figure}

\begin{table}[!t]
\centering
\renewcommand{\arraystretch}{1.15}
\small
\begin{tabular}[c]{|l|c|c|}
\hline
& {\bf Grep} & {\bf WordCount} \\
\hline
\hline
No fault tolerance & 275.83\,MB/s & 149.24\,MB/s \\
\hline
With state backup only & 263.78\,MB/s & 118.41\,MB/s \\
\hline
With item backup only & 136.81\,MB/s & 82.18\,MB/s \\
\hline
With both state and item backups & 98.75\,MB/s & 76.76\,MB/s \\
\hline
\end{tabular}
\caption{Throughput of Flink under different fault tolerance mechanisms on a
	local cluster.} 
\label{tab:flink}
\end{table}

We also study the fault tolerance overhead in Apache Flink (v1.7.2)
\cite{Carbone2015} on a local cluster (see \S\ref{sec:eval} for the cluster
setup).  Flink issues state backups as checkpoints.  For item backups, we
persist items in Kafka \cite{Kafka}, which we configure as a reliable data
source for Flink; if item backups are disabled, we load the whole dataset into
RAM to mitigate the I/O overhead (see Section~\ref{subsubsec:baseline} for
more evaluation on the use of Kafka). 
Table~\ref{tab:flink} shows the throughput of Grep and WordCount with
different fault tolerance approaches.  Enabling both state and item backups
can drop the throughput by 64\% and 49\% in Grep and WordCount, respectively,
compared to without fault tolerance.

\subsection{Classes of Streaming Algorithms}
\label{subsec:classes}

This paper focuses on three classes of streaming algorithms
that have been well studied and widely deployed.

\paragraph{Data synopsis.} Data synopsis algorithms summarize vital
information of large-volume data streams into compact in-memory data
structures with low time and space complexities.  Examples include sampling,
histograms, wavelets, and sketches \cite{Cormode2012}, and have been used in
areas such as anomaly detection in network traffic
\cite{Estan2002,Cormode2004} and social network analysis \cite{Song2009}. 
To bound the memory usage of data structures, data synopsis algorithms are
often designed to return estimates with bounded errors.  For example,
sampling algorithms (e.g., \cite{Estan2002}) perform computations on a subset
of items; sketch-based algorithms (e.g., \cite{Estan2002,Cormode2004}) map a
large key space into a fixed-size two-dimensional array of counters.

\paragraph{Stream database queries.} Stream databases manage data streams with
SQL-like operators as in traditional relational databases, and allow queries
to be continuously executed over continuous streams.  Since some SQL operators
(e.g., join, sorting, etc.) require multiple iterations to process items,
stream database queries need to adapt the semantics of SQL operators for
stream processing.  For example, they restrict the processing of items over
a time window, or return approximate query results using data synopsis
techniques (e.g., sampling in online join queries \cite{Haas1999,Luo2002}).

\paragraph{Online learning.}  Machine learning aims to model the properties of
data by processing the data (possibly over iterations) and identifying the
optimal parameters toward some objective function. It has been widely used in
web search, advertising, and analytics.

Traditional machine learning algorithms assume that the whole dataset is
available in advance and iteratively refine parameters towards a global
optimization objective on the whole dataset.  
To support stream processing, online learning algorithms define a
local objective function with respect to the current model parameter values,
and search for the model parameter values that optimize the local objective
function.  
After a large number of items are processed, it has been
shown that the local approach can converge to a global optimal point, subject
to certain conditions \cite{Zinkevich2003,Hoffman2010,Hoffman2012}.

\begin{table*}[!t]
	\centering
	\renewcommand{\arraystretch}{1.05}
	\small
	\begin{tabular}[c]{|p{1.8in}|p{4.4in}|}
		\hline
		\textbf{Common features}  & \textbf{Corresponding design choices}\\
		\hline
		\hline
		Common primitive operators  & AF-Stream abstracts a streaming algorithm into a
		set of operators and maintains fault tolerance for each operator
		(\S\ref{subsec:arch}).  It also realizes a rich set of built-in primitive
		operators (\S\ref{sec:impl}).\\
		\hline
		Intensive and skewed state updates & AF-Stream supports partial state backup
		to reduce the backup size (\S\ref{subsubsec:state}). It also exposes
		interfaces to let users specify which parts of a state are actually included
		in a backup (\S\ref{subsec:program}).\\
		\hline
		Bounding by windows & AF-Stream resets thresholds based on windows 
		(\S\ref{subsubsec:multiple}).  It also exposes interfaces to let users specify
		window types and lengths (\S\ref{subsec:program}).\\
		\hline
	\end{tabular}
	\caption{Common features of streaming algorithms and their corresponding
		design choices.}
	\label{tab:common}
\end{table*}

\subsection{Common Features}
\label{subsec:commonalities}

We identify the common features of existing streaming algorithms to be
addressed in our system design.  Table~\ref{tab:common} shows how our
design choices are related to the common features. 

\paragraph{Common primitive operators.}  We can often decompose an
operator of a streaming algorithm into a number of {\em primitive
	operators}, which form the building blocks of the same class of streaming
algorithms. For example, stream database queries are formed by few operators
such as map, union, and join \cite{Abadi2003, Gulisano2010, Qian2013};
sketch-based algorithms are formed by hash functions, matching, and numeric
arrays \cite{Yu2013}.

\paragraph{Intensive and skewed updates.}  The state of an operator often
holds {\em update-intensive} values, such as sketch counters \cite{Yu2013} and
model parameters in online learning \cite{Li2014}.  Also, state values are
updated at different frequencies.  For example, the counters corresponding to
frequent items are updated at higher rates; in online learning, only few model
parameters are frequently updated due to the sparsity nature of machine
learning features \cite{Li2014}.

\paragraph{Bounding by windows.}  Streaming algorithms often work on a bounded
sequence of items of a stream.  For example, some operators of stream database
queries process a time window of items (\S\ref{subsec:classes}); some
incremental processing systems (e.g.,
\cite{Logothetis2010,Chen2010,Bhatotia2011,Zaharia2013}) divide
a stream into mini-batches for processing.

\subsection{Distributed Implementation}
\label{subsec:distributed}

We can parallelize a streaming algorithm via distributed implementation.  We
identify two common distributed approaches, which can be used individually or
in combination.

\paragraph{Pipelining.} It divides an operator into multiple stages,
each of which corresponds to an operator or a primitive operator
(\S\ref{subsec:commonalities}).  The output items of one stage can serve as
input items to the next stage, while different stages can process different
items in parallel.

\paragraph{Operator duplication.} It parallelizes stream
processing via multiple copies of the same operator.  There are two ways
to distribute loads across operator copies: {\em data partitioning}, in which
each operator copy processes a subset of items of a stream, and {\em
	state partitioning}, in which each operator copy manages a subset of values of
a state.  Both approaches can be used simultaneously in a streaming algorithm. 

\section{AF-Stream Design}
\label{sec:system}

AF-Stream abstracts a streaming algorithm as a set of operators.  It maintains
fault tolerance for each operator by making backups for its state and
unprocessed items.
To realize approximate fault tolerance, AF-Stream issues a backup operation
only when the actual state and output deviate much from the ideal state and
output, respectively.  This mitigates backup overhead, while incurring bounded
errors after failures are recovered.

AF-Stream's approximate fault tolerance inherently differs from existing
backup-based approaches for distributed stream processing.  Unlike the
approaches that achieve error-free fault tolerance (e.g.,
\cite{Hwang2007,Akidau2013,Qian2013}), AF-Stream issues fewer backups, thereby
improving stream processing performance.  In addition, unlike the approaches
that achieve best-effort fault tolerance by also making fewer backups (e.g.,
\cite{Neumeyer2010,Storm}), AF-Stream ensures that the errors are bounded with
theoretical guarantees.  AF-Stream also differs from existing approximate
processing systems (\S\ref{sec:related}) in that it only applies approximation
in maintaining fault tolerance rather than in normal processing.

\subsection{Design Assumptions}
\label{subsec:assumptions}

To bound the errors upon failures, AF-Stream makes the following assumptions
on streaming algorithms.

AF-Stream assumes that a single lost item (without any backup) brings limited
degradations to accuracy.  Instead of focusing on identifying a specific item
(e.g., finding an outlier item), AF-Stream is designed to analyze the
aggregated behavior over a large-volume stream of items, so each item has
limited impact on the overall analysis.  For example, in network monitoring, 
we may want to identify the flows whose sums of packet sizes exceed a threshold.
Each item corresponds to a packet, whose maximum size is typically limited by
the network (e.g., 1,500 bytes).  Note that this assumption is also made by
existing stream processing systems that build on approximation techniques 
(\S\ref{sec:related}). 

Note that after we recover a failed operator, there may be errors when the
restored operator resumes processing from an actual state instead of from the
ideal state.  Nevertheless, such errors can often be amortized or compensated
after processing a sufficiently large number of items.  For example, online
learning algorithms can converge to the optimal solution even if they
start from a non-ideal state \cite{Langford2009, Smola2010, Ho2013}.  Thus,
this type of errors brings limited accuracy degradations, as also validated by
our experiments (\S\ref{sec:eval}).

Finally, AF-Stream assumes that the errors across multiple duplicate operator
copies can be aggregated. For example, machine learning algorithms maintain
linearly additive states \cite{Li2014}, thereby allowing the errors of
multiple copies to be summable.

\subsection{Architecture}
\label{subsec:arch}

\begin{figure}[!t]
	\centering
	\includegraphics[width=3.2in]{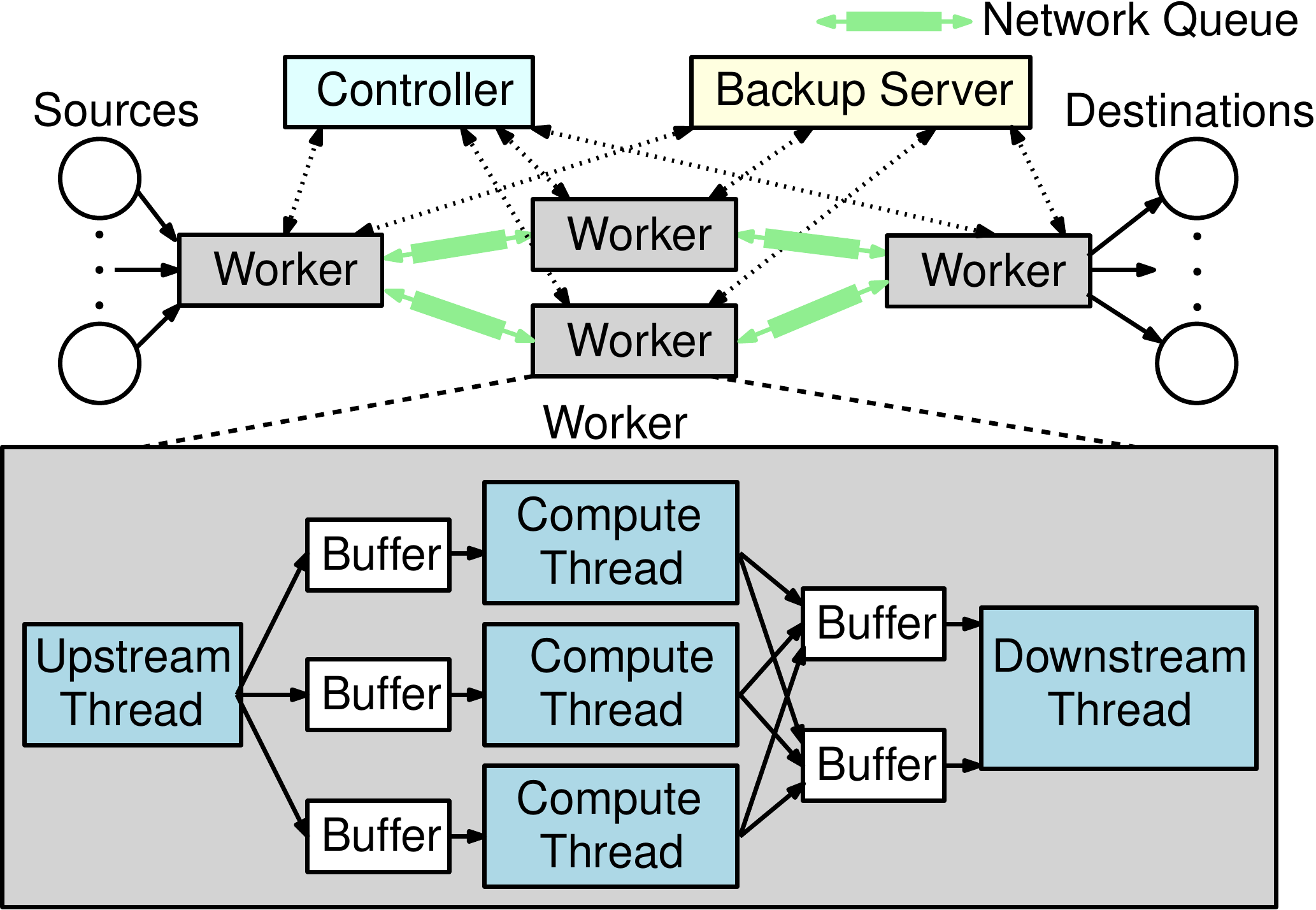}
	\caption{AF-Stream architecture.}
	\label{fig:arch}
	\vspace{-6pt}
\end{figure}

Figure~\ref{fig:arch} shows the architecture of AF-Stream.  AF-Stream
comprises multiple processes, including a single {\em controller} and multiple
{\em workers}.  Each worker manages a single operator of a streaming
algorithm, while the controller coordinates the executions of all workers.
AF-stream organizes workers as a graph, in which one or multiple sources
originate data streams, and one or multiple destinations store the final
results.  For a pair of neighboring workers, say $w_1$ and $w_2$, we call
$w_1$ an {\em upstream} worker and $w_2$ a {\em downstream} worker if the
stream processing directs from $w_1$ to $w_2$.  Specifically, a worker
receives input items from either a source or an upstream worker, processes the
items, and forwards output items to either a destination or a downstream
worker.  

AF-Stream also supports the {\em feedback} mechanism, which is essential for
some streaming algorithms (e.g., model convergence in online learning
\cite{Langford2009}).  It allows a downstream worker to optionally send
feedback messages to an upstream worker.  In other words, the communication
between each pair of neighboring workers is bi-directional.  

The controller manages the execution of each worker, which periodically sends
heartbeats to the controller.  If a worker fails, the controller recovers the
failed state and data in a new worker.  Also, each worker issues backups to a
centralized {\em backup server}, which keeps backups in reliable storage.  
The backup server should be viewed as a {\em logical} entity that can be 
substituted with any external storage system (e.g., HDFS \cite{Shvachko10}).
In this paper, we assume that the controller and the backup server are always
available and have sufficient computational resources, yet we can deploy
multiple controllers and backup servers for fault tolerance and scalability.

Each worker in AF-Stream comprises one upstream thread,
one downstream thread, and multiple compute threads.  The upstream thread
forwards input items from either a source or an upstream worker to one of the
compute threads, while the downstream thread forwards the output items from
the compute threads to either a destination or a downstream worker.  In
particular, multiple compute threads can collaboratively process items, such
that a compute thread can partially process an item and forward the
intermediate results to another compute thread for further processing.
Furthermore, the downstream thread can collect and forward any feedback
message from a downstream worker to the compute threads for processing, and
the upstream thread can forward any new feedback message from the compute
threads to an upstream worker.  Our implementation experience is that it
suffices to have only one upstream thread and one downstream thread per worker
to achieve the required processing performance. Thus, we can reserve the
remaining CPU cores for compute threads to perform heavy-weight computations.

AF-Stream connects workers and threads as follows.  For inter-worker
communications, it connects each pair of upstream and downstream
workers via a bi-directional network queue.  For inter-thread communications,
it shares data across threads via in-memory circular ring buffers.
We carefully optimize both network queue and ring buffer implementations so as
to mitigate the communication overhead (\S\ref{sec:impl}).

\subsection{Programming Model}
\label{subsec:program}

AF-Stream manages two types of objects: {\em states} and {\em items}
(\S\ref{subsec:abstraction}), whose formats are user-defined.
Each operator is associated with a state\footnote{We also allow an operator to
	maintain an empty state (i.e., stateless).}, and each state holds an array of
binary values.  Also, each operator supports three types of items:
(i) {\em data items}, which collectively refer to the input and output items
that traverse along workers from upstream to downstream in stream processing,
(ii) {\em feedback items}, which traverse along workers from downstream to
upstream, and (iii) {\em punctuation items}, which specify the end of an
entire stream or a sub-stream for windowing (\S\ref{subsec:commonalities}).

AF-Stream has two sets of interfaces: {\em composing interfaces} and
{\em user-defined interfaces}, as listed in 
Appendix~\ref{sec:appendix_interfaces}. 
The composing interfaces are used to define the AF-Stream architecture and
the stream processing workflows.  Their functionalities can be
summarized as follows: (i) connecting workers (in server hostnames) and the
source/destination (in file pathnames), (ii) adding threads to each worker,
(iii) connecting threads within each worker, (iv) pinning a thread to a CPU
core, and (v) specifying the windowing type (e.g., hopping window, sliding
window, decaying window) and window length.

On the other hand, user-defined interfaces allow programmers to add specific
implementation details.  AF-Stream automatically calls the user-defined
interfaces and processes items based on their implementations.  Specifically,
the upstream thread can be customized to receive items, dispatch them to the
compute threads, and optionally send feedback items to upstream workers.
Similarly, the downstream thread can be customized to send output items and
optionally receive feedback items from downstream workers.  The compute
thread can be customized to process data items, feedback items, and
punctuation items.

AF-Stream provides three user-defined interfaces for building
operators that realize approximate fault tolerance for state
backup (\S\ref{subsubsec:state}): (i) {\tt StateDivergence}, which quantifies
the divergence between the current state and the most recent backup state,
(ii) {\tt BackupState}, by which operators provide the state to be saved in
reliable storage via the backup server, and (iii) {\tt RecoverState}, by which
operators obtain the most recent backup state from the backup server.
Such interfaces are included in a base class called 
{\tt FTOperator}, which can be extended by specific fault-tolerant operators.
By default, the three interfaces have empty implementations, meaning that
state backup is disabled.

\begin{lstlisting}[frame=lines, language=C++, caption=Implementation of the fault-tolerant hashmap., captionpos=b, label=lst:example, float=!t]template<class K, class V> 
class FTHashMap: public FTOpertor {
private:
	std::unordered_map<K,V> map;  // current map
	std::unordered_map<K,V> bkup_map; // backup map
public:
	double StateDivergence(){
	double d = 0;
		for (auto it = map.begin(); it != map.end(); it++)
			d += abs(map[it->first] - bkup_map[it->first]);
		return d;	
	}
	void BackupState(BackupData &backup_data){
		backup_data.meta.len = sizeof(map);
		memcpy(backup_data.data, &map, sizeof(map));
		bkup_map = map;
		// backup_data is sent to backup server
	}
	void RecoverState(BackupData &recover_data){
		// recovery_data is fetched from backup server
		memcpy(&map, recover_data.data, sizeof(map));
		bkup_map = map;
	}
};
\end{lstlisting}

For example, Listing~\ref{lst:example} shows the code snippet of a
fault-tolerant hashmap {\tt FTHashMap}, which extends the base class 
{\tt FTOperator}. In
the implementation, {\tt StateDivergence} can specify the divergence of the
hash map (e.g., in Listing~\ref{lst:example}, we use the Manhattan distance of
counter values between the current and backup hash maps).  {\tt BackupState}
can issue backups for the counter values in the hashmap, while 
{\tt RecoverState} can restore the hashmap.  Here, we assume that the backup
hashmap is stored in a backup server.  Note that {\tt BackupState} and 
{\tt RecoverState} also update the current backup hashmap in {\tt bkup\_map}
for the computation of the state divergence in {\tt StateDivergence}.  We
further elaborate how the functions are called in \S\ref{subsubsec:state}.

\subsection{Approximate Fault Tolerance}
\label{subsec:ft}

AF-Stream maintains approximate fault tolerance for both the state and items
of each operator.  We introduce both state backup (\S\ref{subsubsec:state})
and item backup (\S\ref{subsubsec:item_backup}) as individual backup
mechanisms that are complementary to each other and are configured by
different thresholds. 

\subsubsection{State Backup}
\label{subsubsec:state}

Recall that AF-Stream calls {\tt BackupState} to issue a backup operation for
the state of an operator to the backup server.  Instead of making frequent
calls to {\tt BackupState}, AF-Stream defers the call to {\tt BackupState}
until the current state deviates from the most recent backup state by some
threshold (denoted by $\theta$).  Specifically, AF-Stream caches a copy of the
most recent backup state of the operator in local memory.  Each time when an
item updates the current state, AF-Stream calls {\tt StateDivergence} to
compute the divergence between the current state and the cached backup state.
If the divergence exceeds $\theta$, AF-Stream calls {\tt BackupState} to issue
a backup for the current state, and updates the cached copy accordingly.

Making the backup of an entire state may be expensive, especially if the state
is large.  To further mitigate backup overhead, AF-Stream supports the
backup of a {\em partial} state, such that programmers can only specify the
list of updated values (together with their indices) that make a state
substantially deviate as the returned
state of {\tt BackupState}.  Partial state backup is useful when only few
values of a state are updated (\S\ref{subsec:commonalities}).

Each state update triggers a call to {\tt StateDivergence}.  For
most operators and divergence functions (e.g., difference of cardinalities,
Manhattan distance, Euclidean distance, or maximum difference of values),
{\tt StateDivergence} only involves
few arithmetic operations.  For example, suppose that a divergence function
returns the sum of differences (denoted by $D$) of all values between two
states.  If we update a value, then we can compute the new sum of differences
(denoted by $D'$) as $D' = D + \Delta$, where $\Delta$ is the change of the
value.  Thus, if the operator has only one compute thread, the compute thread
itself can evaluate the divergence with limited overhead.  On the other hand,
if the operator has multiple compute threads to process a state simultaneously,
say by operator duplication (\S\ref{subsec:distributed}), then summing the
divergence of all compute threads can be expensive due to inter-thread
communications.  AF-Stream currently employs a smaller threshold in each
compute thread, such that the sum of the thresholds does not exceed $\theta$.
For example, we can set the threshold as $\theta/n$ if there are $n > 1$
compute threads.

\subsubsection{Item Backup}
\label{subsubsec:item_backup}

AF-Stream makes item backups {\em selective}, such that a worker decides to
make a backup for an item based on the item type and a pre-defined threshold.
For a punctuation item, a worker always makes a backup for it to
accurately identify the head and tail of a sequence to be processed; for a
data item or a feedback item, the worker counts the number of pending items
that have not yet been processed.  If the number exceeds a threshold (denoted
by $l$), the worker makes backups for all pending items so as to bound the
number of missing items upon failures.

In AF-Stream, item backups are issued by the receiver-side worker (i.e., a
downstream worker handles the backups of data and punctuation items, while an
upstream worker handles the backups of feedback items).  The rationale is that
the receiver-side worker can exactly count the number of unprocessed items and
decide when to issue item backups.  Consider a data item that traverses from
an upstream worker to a downstream worker.  After the upstream worker sends
the data item, it keeps the data item in memory.  When the downstream worker
receives the data item, it examines the number of unprocessed data items.  If
the number exceeds the threshold $l$, then the worker makes backups for all
pending data items before dispatching the data item to the compute threads for
processing.  It also returns an acknowledgement (ACK) to the upstream worker,
which can then release the cached data item from memory.

Performing item backups on the receiver side enables us to handle ACKs
differently from the traditional upstream backup approach 
\cite{Akidau2013, Hwang2007, Qian2013}, as shown in
Figure~\ref{fig:data_backup}.  Upstream backup makes item backups on the
sender side.  For example, an upstream worker is responsible for making item
backups for data items. It caches the output items in memory and waits until
its downstream worker replies an ACK.  However, the downstream worker sends
the ACK only when the item is completely processed, and the upstream worker
may need to cache the item for a long time.  In contrast, our receiver-side
backup approach can send an ACK before an item is processed, and hence limits
the caching time in the upstream worker.

We emphasize that our receiver-side backup approach should be viewed as a
complementary solution (rather than a replacement) to upstream backup from the
perspective of approximate fault tolerance.  It still caches items in memory
in the upstream worker (as shown in Figure~\ref{fig:data_backup}(b)) as in
upstream backup, but opts to issue backups on the receiver side rather than on
the sender side.  As discussed above, it not only reduces the caching time in
the upstream worker, but also favors our approximate fault tolerance design as
the receiver-side worker can count the number of unprocessed data items and
issue item backups if the number exceeds the threshold $l$.

\begin{figure}[!t]
	\centering
	\includegraphics[width=3in]{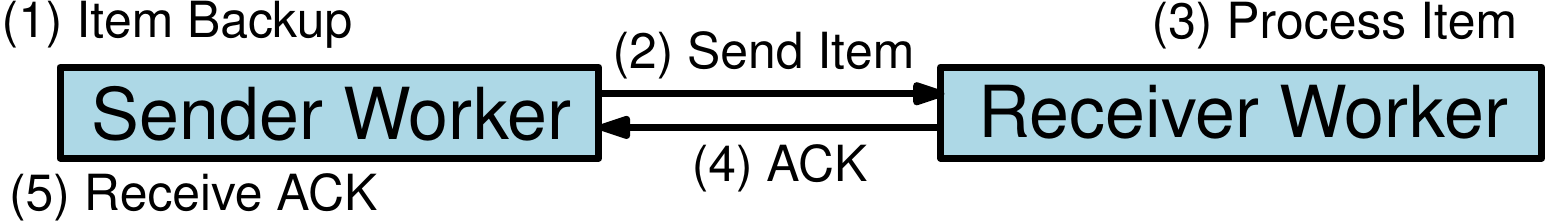}\\
	{\small (a) Upstream backup} \\
	\vspace{1em}
	\includegraphics[width=3in]{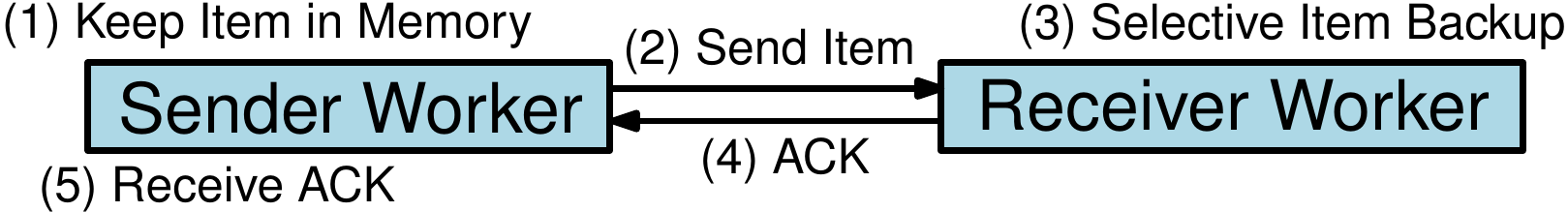}\\
	{\small (b) Our receiver-side backup}
	\caption{Comparisons between upstream backup and our receiver-side backup.}
	\label{fig:data_backup}
	\vspace{-6pt}
\end{figure}

The rate of item backups in AF-Stream is responsive to the current load of a
worker.  Specifically, when a worker is heavily loaded, it will accumulate
more unprocessed items and hence trigger more backups.  Nevertheless, the
backup overhead remains limited since item processing is the bottleneck in
this case.  On the other hand, when a worker is lightly loaded, it issues
fewer backups and avoids compromising the processing performance.

Note that the sender-side worker may miss all unacknowledged items if it
fails.  Currently AF-Stream does not make backups for such unacknowledged
items, but instead bounds the maximum number of unacknowledged items (denoted
by $\gamma$) that a worker can cache in memory.  If the number of
unacknowledged items reaches $\gamma$, then the worker is blocked from sending
new items until the number drops.  The value of $\gamma$ can be generally very
small as our receiver-side approach allows a worker to reply an ACK
immediately upon receiving an item.

\subsubsection{Recovery}
\label{subsubsec:recovery}

If the controller detects a failed worker, it activates recovery and restores
the state of the failed worker in a new worker.  The new worker calls
{\tt RecoverState} to retrieve the most recent backup state.  Also, AF-Stream
replays the backup items into the new worker, which can then process them to
update the restored state.

Some backup items may have already been processed and updated in the recovered
state, and we should avoid the duplicate processing of those items.  AF-Stream
uses the sequence number information for backup and recovery.  When a worker
receives an item, it associates the item with a sequence number.  Each
state also keeps the sequence number of the latest item that it includes.  When
AF-Stream replays backup items during recovery, it only chooses the items
whose sequence numbers are larger than the sequence number kept in the
restored state.
To reduce the recovery time, our current implementation restores fewer items
by replaying only the most recent sequence of consecutive items. Since the
number of unprocessed items before the replayed sequence is at most $l$, the
errors are still bounded.

\subsubsection{User-Configurable Thresholds}
\label{subsubsec:config_thresholds}

AF-Stream exports three user-configurable threshold parameters: (i) $\Theta$,
the maximum divergence between the current state and the most recent backup
state, (ii) $L$, the maximum number of unprocessed non-backup items and (iii)
$\Gamma$, the maximum number of unacknowledged items.  It automatically 
{\em tunes} the thresholds $\theta$, $l$, and $\gamma$ at runtime with respect
to $\Theta$, $L$ and $\Gamma$, respectively, such that the errors are bounded
independent of the number of failures that have occurred and the number of
workers in a distributed environment.  

How the error bounds quantify the accuracy of a streaming algorithm is
specific to the algorithmic design and cannot be directly generalized for
all streaming algorithms.  Nevertheless, in Appendix~\ref{sec:appendix_error}, 
we present both theoretical analysis and numerical examples on how these
parameters are translated to the accuracy of two streaming algorithms, namely
heavy-hitter detection in Count-Min Sketch \cite{Cormode2005a} and Ripple Join
\cite{Haas1999}, and show how their accuracy will be degraded under
approximate fault tolerance with respect to $\Theta$, $L$, and $\Gamma$.
Furthermore, in \S\ref{sec:theory}, we formally analyze how these parameters
affect the model convergence of online learning.  We also resort to
experiments to empirically evaluate the accuracy for various parameter choices
in AF-Stream (\S\ref{sec:eval}).

AF-Stream currently requires users to have domain knowledge on configuring the
parameters with respect to the desired level of accuracy.  We pose the issue
of configuring the parameters without user intervention as future work.

\subsection{Error Analysis}
\label{subsec:ft_analysis}

We analyze how AF-Stream bounds the errors upon failures for different
failure scenarios in a distributed environment.
We quantify the error bounds in two aspects: (i) the divergence between the
actual and ideal states and (ii) the number of lost output items, based on
our assumptions (\S\ref{subsec:assumptions}).
In particular, we assume that each lost item brings limited accuracy
degradations.  To quantify the degradations, after updating the current state
with an item, we let the divergence between the current state and the most
recent backup state increase by at most $\alpha$ and the number of output
items generated by the update be at most $\beta$, where $\alpha$ and $\beta$
are two constants specific for a streaming algorithm.  Note that both
$\alpha$ and $\beta$ are introduced purely for our analysis, while programmers
only need to configure $\Theta$, $L$, and $\Gamma$ for building a streaming
algorithm.  Also, our analysis does not assume any probability distribution of
failure occurrences.

\subsubsection{Single Failure of a Worker}
\label{subsubsec:single}

\begin{figure}[t]
	\centering
	\includegraphics[width=3in]{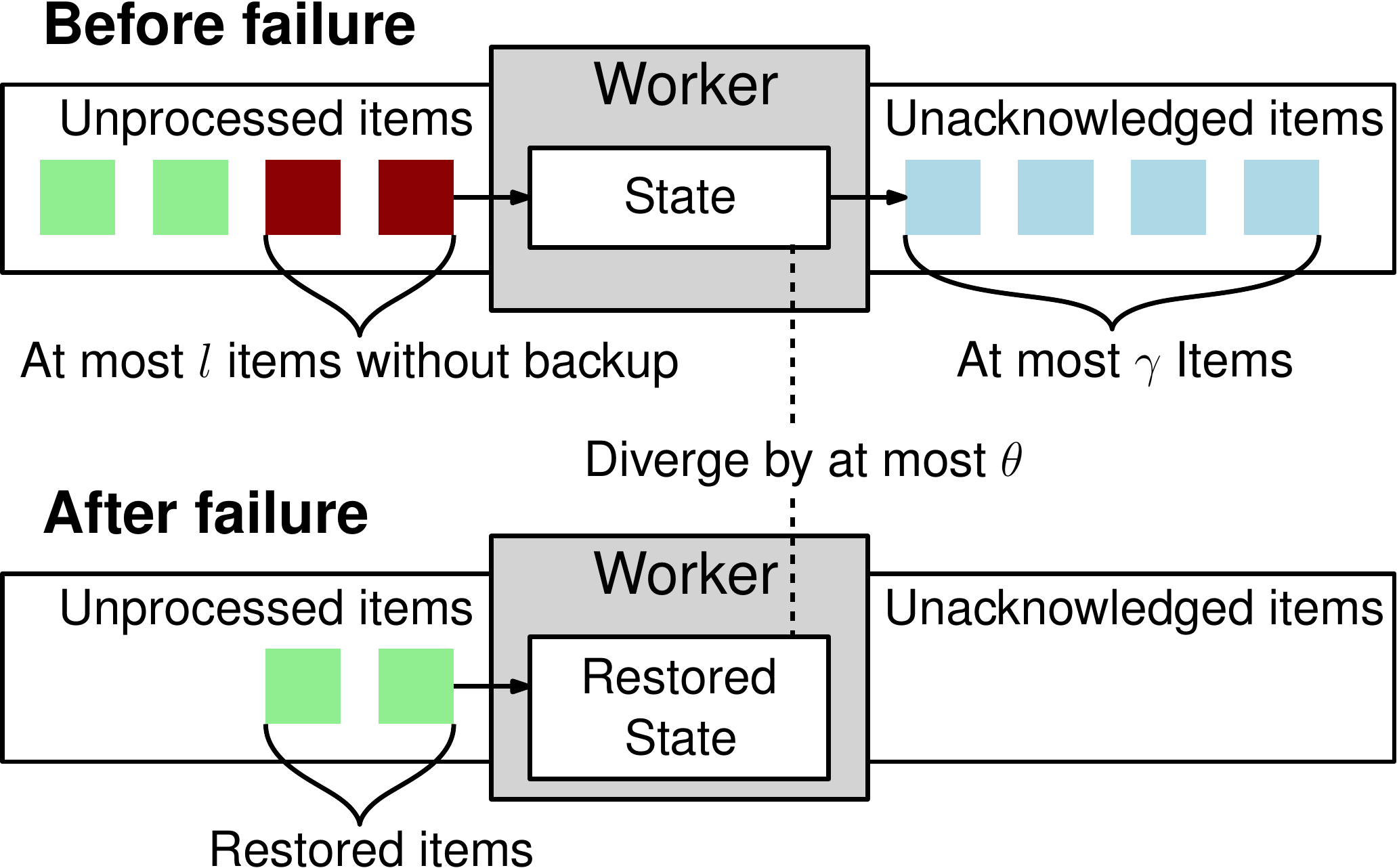}
	\caption{Errors due to a single failure of a worker.}
	\label{fig:fail}
	\vspace{-6pt}
\end{figure}

Suppose that AF-Stream sees only a single failure of a worker over its
lifetime.  We consider the worst case when the failed worker is restored, as
shown in Figure~\ref{fig:fail}.  First, for the divergence between the actual
and ideal states, the restored state diverges from the actual state before the
failure by at most $\theta$, and each of the lost unprocessed items (with $l$
at most) changes the divergence by at most $\alpha$.  Thus, the total
divergence between the actual and ideal states is at most $\theta + l\alpha$,
which is bounded.  Second, for the number of lost output items, AF-Stream
loses at most $\gamma$ unacknowledged output items, and each of the lost
unprocessed items (with $l$ at most) leads to at most $\beta$ lost output
items.  Thus, the total number of lost output items is at most $\gamma +
l\beta$, which is also bounded.  By initializing $\theta$, $l$, and $\gamma$
as $\Theta$, $L$, and $\Gamma$, respectively, we can bound the errors with
respect to the user-configurable thresholds.

\subsubsection{Multiple Failures of a Single Worker}
\label{subsubsec:multiple}

Suppose that AF-Stream sees multiple failures of a single worker over its
lifetime, while other workers do not fail.  In this case, the errors after
each failure recovery will be accumulated.  AF-Stream bounds the accumulated
errors by adapting the three thresholds $\theta$, $l$, and $\gamma$ with
respect to the three user-configurable parameters $\Theta$, $L$ and $\Gamma$,
respectively, and the number of failures denoted by $k$.  First, a worker
initializes the thresholds with $\theta=\frac{\Theta}{2}$, $l=\frac{L}{2}$,
and $\gamma=\frac{\Gamma}{2}$.  After each failure recovery, the worker halves
each threshold; in other words, after recovering from $k$ failures, the
thresholds become $\theta=\frac{\Theta}{2^{k+1}}$,
$l=\frac{L}{2^{k+1}}$, and
$\gamma=\frac{\Gamma}{2^{k+1}}$.
By summing up the errors accumulated over $k$ failures, 
we can show that the divergence between the actual
and ideal states is at most $\Theta+L\alpha$ and the number of
lost output items is at most $\Gamma+L\beta$.  Note that the errors are
bounded independent of the number of failures $k$.

Our adaptations imply that the thresholds become very small after many
failures, so AF-Stream reduces to error-free fault tolerance and makes
frequent backups.  AF-Stream addresses this issue by allowing the
thresholds to be reset.  In particular, many streaming algorithms work with
windowing, and AF-Stream can use different strategies to reset the thresholds
for different windowing types (\S\ref{subsec:program}).  For example,
for the hopping window, AF-Stream resets the thresholds to the initial values
$\theta\!=\!\frac{\Theta}{2}$, $l\!=\!\frac{L}{2}$, and
$\gamma\!=\!\frac{\Gamma}{2}$ at
the end of a window; for the sliding window, AF-Stream tracks the time of the
last failure and increases the thresholds once a failure is not included in
the window;  for the decaying window, AF-Stream always keeps
$\theta\!=\!\Theta$, $l\!=\!L$, and $\gamma\!=\!\Gamma$ and disables threshold
adaptations, as the errors fade over time.

\subsubsection{Failures in Multiple Workers}
\label{subsubsec:multiple2}

We address the general case when a general number of failures can happen
in multiple workers in a distributed environment.  To make the error bounds
independent of the number of workers, AF-Stream employs small initial
thresholds, such that the accumulated errors are the same as those in the
single-worker scenario.  Specifically, for operator duplication with $n$
copies, AF-Stream initializes each copy with $\theta\!=\!\frac{\Theta}{2n}$,
$l\!=\!\frac{L}{2n}$, and $\gamma\!=\!\frac{\Gamma}{2n}$; for a pipeline with
$m$ operators, AF-Stream initializes the $i$-th operator in the pipeline with
$\theta\!=\!\frac{\Theta}{2\beta^{m-i}}$, $l\!=\!\frac{L}{2\beta^{m-i}}$, and
$\gamma\!=\!\frac{\Gamma}{2\beta^{m-i}}$ (since each lost item can lead to at
most $\beta^m$ lost output items after $m$ pipeline stages).  By summing the
errors over all failures and all workers, we can obtain the same error bounds
as in \S\ref{subsubsec:multiple}.

\section{Case Study: Convergence of Online Learning}
\label{sec:theory}

In this section, we analyze how AF-Stream affects the {\em model convergence}
of online learning.  

Our insight is that general machine learning algorithms are
known to be iterative-convergent and error-tolerant \cite{Bottou2007,Dai2015},
meaning that they still converge to an optimum even in the face of state
errors and loss of items.  Offline parameter server systems
\cite{Dai2015,Ho2013,Li2014} have leveraged this nature and proposed relaxed
consistency models to improve machine learning performance.   Here, we prove
that AF-Stream preserves the iterative-convergent and error-tolerant nature of
online learning independent of the number of failures $k$ (defined in
\S\ref{subsubsec:multiple}). 

We focus on online learning algorithms based on {\em stochastic gradient
	descent (SGD)} \cite{Bottou1998}, which has been used by many popular machine
learning algorithms to learn model parameters via stochastic optimization
\cite{Bottou2010}.  Our goal is to prove that upon the recovery of 
{\em multiple} failures, online learning algorithms based on SGD still
converge to an optimum under AF-Stream. 

Note that our analysis inherently differs from the standard SGD analysis, as
we must address the complicated interactions of SGD and the approximate fault
tolerance design of AF-Stream.  In particular, after each failure recovery,
AF-Stream will (slightly) deviate SGD from its standard behavior, and
characterizing the aggregated impact of approximate fault tolerance after
multiple failures is non-trivial. 

We formulate our analysis as follows. 
SGD maintains its state as a vector of model parameters.  Let $\PS_t$ denote
the state vector of SGD after the $t$-th item $X_t$ is processed,
assuming that no failure happens.  SGD defines a risk function $f_t$ with
respect to the state for each $X_t$.  To update model parameters, SGD first
computes the gradient of $\Pg_t = \nabla f_t(\PS_{t-1})$, followed by updating
the vector of model parameters $\PS_{t} = \PS_{t-1}-\eta_t\Pg_t$, where
$\eta_t=\tfrac{1}{\sqrt{t}}$ is a parameter that controls the learning rate.

The convergence theory of SGD \cite{Langford2009,Zinkevich2003} further
defines the average risk function (also known as {\em regret}) of the first
$T$ items as $\tfrac{1}{T}\sum_{t=1}^{T}f_t(\PS_t)$.  The theory addresses the
problem whether the regret can be minimized.  Prior studies
\cite{Dai2015,Ho2013,Langford2009} prove that under certain conditions the
regret converges to the global optimum (denoted by $\PS^*$):
\begin{equation}
\lim\limits_{T\to\infty}(\AvgT f_t(\PS_t)-\AvgT f_t(\PS^*))=0.
\nonumber
\end{equation}

We denote the actual model parameters by $\NS_t$, as opposed to ideal model
parameters denoted by $S_t$.  We also denote the global optimum by $\PS^*$.
Let $k$ be the number of failures encountered by AF-Stream throughout its
lifetime.  In practice, we expect that $k$ is much less than $T$ (i.e., the
number of items being processed), since failures are less common compared to
the numerous items in stream processing.  Upon detecting each failure,
AF-Stream performs recovery (\S\ref{subsubsec:recovery}).
Theorem~\ref{theorem:sgd} formalizes the convergence guarantee. 
\begin{theorem}
	\label{theorem:sgd}
	Consider an SGD sequence $\{\NS_t\}$ generated under AF-Stream.
	The limit
	\begin{equation}
	\lim\limits_{T\to\infty}[\AvgT f_t(\NS_t) - \AvgT f_t(\PS^*)]=0 \nonumber
	\end{equation}
	still holds even after the recovery of $k$ failures, provided that $k \ll T$. 
\end{theorem}
\noindent
{\bf Proof.} See Appendix~\ref{sec:appendix_convergence}. 

\paragraph{Remarks.} Theorem~\ref{theorem:sgd} implies that the model
convergence of AF-Stream still holds regardless of the value of $k$, provided
that $k \ll T$.  We also resort to experiments to empirically evaluate the
convergence performance of AF-Stream for various settings
(\S\ref{subsec:eval_convergence}).

\section{Implementation}
\label{sec:impl}

We have implemented a prototype of AF-Stream in C++.  AF-Stream connects the
controller and all workers using ZooKeeper \cite{Hunt2010} to manage fault
tolerance.  It also includes a backup server, which is now implemented
as a daemon that receives backup states and items from the workers via TCP
and writes them to its local disk.  Our current prototype has around 47,800
LOC in total.

Our prototype realizes a number of built-in primitive operators and their
implementations of {\tt StateDivergence}, {\tt BackupState}, and {\tt
	RecoverState}, listed in Appendix \ref{sec:appendix_interfaces}.  We now
support the numeric variable, vector, matrix, hash table, and set, and provide
them with built-in fault tolerance.  For example, we can keep elements in our
built-in hash table, whose fault tolerance is automatically enabled.
Programmers can also build their own fault-tolerant operators via implementing
the above three interfaces.  

While we currently implement AF-Stream as a clean-slate system, we explore how
to integrate our approximate fault tolerance notion into existing stream
processing systems in future work.  Such integration must address at least the
following technical challenges, including: (i) the three interfaces 
(i.e., {\tt StateDivergence}, {\tt BackupState}, and {\tt RecoverState}) for
state backup (\S\ref{subsubsec:state}), (ii) the receiver-side backup and
acknowledgement mechanisms for item backup (\S\ref{subsubsec:item_backup}), 
(iii) the new recovery mechanism (\S\ref{subsubsec:recovery}), and (iv) the
user-configurable thresholds (\S\ref{subsubsec:config_thresholds}).

\paragraph{Communication optimization.} 
Our prototype specifically optimizes both inter-thread and
inter-worker communications for high-throughput stream processing.  For
inter-thread communications, we implement a lock-free multi-producer,
single-consumer (MPSC) ring buffer \cite{NatSysLab}.  We assign one MPSC ring
buffer per destination, and group output items from different compute threads
by destinations (\S\ref{subsec:arch}).  This offloads output item scheduling
to compute threads, and simplifies subsequent inter-worker communications.
Also, we only pass the pointers to the items to the ring buffer, so that the
rate of the number of items that can be shared is independent of the item
size.

For inter-worker communications, we implement a bi-directional network queue
with ZeroMQ \cite{ZeroMQ}.  ZeroMQ itself uses multiple threads to
manage TCP connections and buffers, and the thread synchronization is
expensive. Thus, we modify ZeroMQ to remove its thread layer, and make the
upstream and downstream threads of a worker directly manage TCP connections
and buffers.  The modifications enable us to achieve high-throughput stream
processing in a 10~Gb/s network (\S\ref{sec:eval}).

\paragraph{Asynchronous backups.} 
We further mitigate the performance degradation of making backups
using an asynchronous technique similar to \cite{Fernandez2014}.
Specifically, our prototype employs a dedicated backup thread for each worker.
It collects all backups of a worker and sends them to the backup server,
allowing other threads to proceed normal processing after generating backups.
Note that the asynchronous technique does not entirely eliminate backup
overhead, since the backup thread can still be overloaded by frequent backup
operations.  Thus, we propose approximate fault tolerance to limit the number
of backup operations. 

\paragraph{Consistency models.}  AF-Stream currently supports different
consistency models to synchronize the views of workers.  For example, in SGD, 
AF-Stream processes streaming items in multiple local workers, each of
which computes gradients on the items based on its local model and sends
the local gradient results to a global worker.  The global worker aggregates
all the results to update its global model, and feedbacks the updated global
model to each local worker to update its local model.   AF-Stream can choose
one of the following three consistency models for synchronizing the model
updates: (i) the Asynchronous Parallel (ASP) model \cite{Zinkevich2010}, in
which after a local worker sends local results to the global worker, it
continues to use its (stale) local model to compute gradients for the new
incoming items, until the global worker feedbacks an updated global model;
(ii) the Bulk Synchronous Parallel (BSP) model \cite{Bradley2011}, in which
after a local worker sends local results to the global worker, it stops
working and waits until receiving the feedback from the global worker; and
(iii) the Stale Synchronous Parallel (SSP) model \cite{Ho2013}, which allows
local workers to use the stale model with at most an additional number of
iterations (which is user-configurable).   Note that the BSP model is adopted
by Spark Streaming \cite{Zaharia2013}, which processes streaming items in
mini-batches.

\section{Experiments}
\label{sec:eval}

We present evaluation results on AF-Stream through experiments on Amazon EC2
and a local cluster.

\begin{table*}[!ht]
\centering
\renewcommand{\arraystretch}{1.0}
\small
\begin{tabular}[c]{|p{0.7in}|p{1.25in}|p{1.25in}|p{1.15in}|p{0.45in}|l|l|}
\hline
\multirow{2}{*}{\bf Algorithm}  
& \multicolumn{1}{c|}{\multirow{1}{*}{{\bf Upstream stage}}}
& \multicolumn{1}{c|}{\multirow{1}{*}{{\bf Downstream stage}}}
& \multicolumn{1}{c|}{\multirow{1}{*}{\bf State}} 
& \multicolumn{3}{c|}{{\bf Trace}}\\
\cline{5-7}
&  \multicolumn{1}{c|}{{\bf (\# workers)}} 
&  \multicolumn{1}{c|}{{\bf (\# workers)}}
& \multicolumn{1}{c|}{{\bf (Divergence)}}
& {\bf Source} & {\bf \# items} & {\bf Size}\\
\hline
\hline
Grep & Parse and send the lines with the matched pattern (10)
	 & Merge matched lines (1) 
	 & None
	 & Gutenberg \cite{Gutenberg} & 300M & 15 GB\\
\hline
WordCount & Parse and send words with intermediate counts (4)
		  & Aggregate intermediate word counts (7)
& Hash table of word counts (maximum difference of word counts)
	& Gutenberg \cite{Gutenberg} & 300M & 15 GB\\
\hline
HH detection
	& Update packet headers into a local sketch and send the local sketch
	and local HHs (10)
	& Merge local sketches into a global sketch, and check if local HHs
	are actual HHs with the global sketch (1)
	& Matrix of counters (maximum difference of counter values in bytes)
	& Caida \cite{Caida} & 1G & 40 GB\\
\hline
Online join
	& Find and send tuples that have matching packet headers in two
	streams (2)
	& Perform join and aggregation (9)
	& Set of sampled items (difference of number of packets)
	& Caida \cite{Caida} & 1G & 40 GB\\
\hline 
Online logistic regression
	& Train and send the local model, and update the
	  local model with feedback items (10)
	& Merge local models to form a global model, and send the global model
	  as feedback items (1)
	& Hash table of model parameters (Euclidean distance)
	& KDD Cup 2012 \cite{Niu2012} & 110M & 42 GB\\
\hline
\end{tabular}
\caption{Summary of streaming algorithms evaluated in \S\ref{subsec:eval_ec2}.}
\label{tab:eval}
\end{table*}

\subsection{Performance on Amazon EC2}
\label{subsec:eval_ec2}

We first evaluate AF-Stream on an Amazon EC2 cluster.  The cluster is located
in the {\tt us-east-1b} zone. It comprises a total of 12 instances: one 
{\tt m4.xlarge} instance with four CPU cores and 16\,GB RAM, and 11 
{\tt c3.8xlarge} instances with 32 CPU cores and 60\,GB RAM each.  We deploy
the controller and the backup server in the {\tt m4.xlarge} instance, and a
worker in each of the {\tt c3.8xlarge} instances.  We connect all instances
via a 10\,Gb/s network.

Our experiments consider five streaming algorithms: Grep, WordCount, heavy
hitter detection, online join, and online logistic regression.  The latter
three algorithms are chosen as the representatives for data synopsis, stream
database queries, and online learning, respectively (\S\ref{subsec:classes}).
We pipeline each streaming algorithm in two stages, in which the upstream
stage reads traces, processes items, and sends intermediate outputs to the
downstream stage for further processing.  Each stage contains one or multiple
workers.  We evenly partition a trace into subsets and assign each subset to
an upstream worker, which partitions its intermediate outputs to different
downstream workers if more than one downstream worker is used.
Table~\ref{tab:eval} summarizes each algorithm, including the functionalities
and number of workers for both upstream and downstream stages, the definitions
of the state and the corresponding state divergence, as well as the source,
number of items, and size of each trace.  We elaborate the algorithm details
later when we present the results.

Each experiment shows the average results over 20~runs.  Before each
measurement, we load traces into the RAM of each upstream worker,
which then reads the traces from RAM during processing.  This eliminates the
overhead of reading on-disk traces, and moves the bottleneck to AF-Stream
itself.  

We evaluate the throughput and accuracy of AF-Stream for each algorithm, and
show the trade-off with respect to $\Theta$, $L$, and $\Gamma$.  For
throughput, we measure the rate of the amount of data processed in each
upstream worker, and compute the sum of rates in all upstream workers as the
resulting throughput.  For accuracy, we provide the specific definition for
each algorithm when we present the results.

\subsubsection{Baseline Performance}
\label{subsubsec:baseline}

We benchmark the baseline performance of AF-Stream using two algorithms: Grep,
which returns the input lines that match a pattern, and WordCount, which
counts the occurrence frequency of each word.  Note that Grep does not need
any state to be kept in a worker, while WordCount defines a state as a hash
table of word counts.  We implement both algorithms based on their
implementations in open-source Spark Streaming \cite{Zaharia2013}.  We use the
documents on Gutenberg \cite{Gutenberg} as the inputs, with the total size
15\,GB.

\paragraph{Experiment~1 (Comparisons with existing fault tolerance
	approaches).} We compare AF-Stream with two open-source stream processing
systems: Storm (v0.9.5) \cite{Storm} and Spark Streaming (v1.4.1) 
\cite{Zaharia2013}.  We
consider three setups for each of them.  The first setup disables fault
tolerance to provide baseline results.  The second setup uses their own
available fault tolerance mechanisms.  Specifically, Storm tracks every item
until it is fully processed (via a component called {\tt Acker}), yet it only
achieves best-effort fault tolerance as it does not support state
backups.  On the other hand, Spark Streaming achieves error-free fault
tolerance by making state backups as RDDs in mini-batches and making item
backups via write-ahead logging \cite{SparkZeroLoss}.  We set the mini-batch
interval of Spark Streaming as 1s (\S\ref{subsec:abstraction}).  For
both setups, we configure the systems to read traces from memory to avoid the
disk access overhead.  Finally, the third setup achieves fault tolerance
through Kafka \cite{Kafka}, a reliable messaging system that persists items to
disk for availability.  We configure the systems to read input items from the
disk-based storage of Kafka, so that Kafka serves as an extra item backup
system to replay items upon failures.  We configure the fault tolerance
mechanisms of Storm and Spark Streaming based on the official documentations;
for Kafka integration, we refer to \cite{KafkaSpark,KafkaStorm}.  In
particular, Kafka and write-ahead logging in Spark Streaming have the same
functionality, so for performance efficiency, it is suggested to disable
write-ahead logging when Kafka is used \cite{KafkaSpark}.  On the other hand,
other fault tolerance approaches can be used in conjunction with Kafka. We
present all combinations in our results. 

In addition, we implement existing fault tolerance approaches in AF-Stream and
compare them with approximate fault tolerance under the same implementation
setting.  We consider {\em mini-batch} and {\em upstream backup}.  Mini-batch
follows Spark Streaming \cite{Zaharia2013} and makes backups for mini-batches.
To realize the mini-batch approach in AF-Stream, we divide a stream into
mini-batches with the same number of items, such that the number of items per
mini-batch is the maximum number while keeping the system stable
(\S\ref{subsec:abstraction}).  Our approach generates around 800 mini-batches.
Our modified AF-Stream then issues state and item backups for each mini-batch.
On the other hand, upstream backup \cite{Hwang2005} provides error-free fault
tolerance by making backups in upstream workers.  To realize upstream backup,
we issue a backup for every data item, while we issue a state backup every 1\%
of data items.  In addition to saving items via the backup server, both
approaches also keep the items in memory for ACKs.  Once the memory usage
exceeds a threshold (1\,GB in our case), we save any new item to local disk.
Finally, we set $\Theta\!=\!10^4$, $L\!=\!10^3$, and $\Gamma\!=\!10^3$ in
AF-Stream for approximate fault tolerance.

\begin{table}[!t]
	\centering
	\renewcommand{\arraystretch}{1.15}
	{\small
		\begin{tabular}[c]{|l|c|c|}
			\hline
			& {\bf Grep}
			& {\bf WordCount} \\
			\hline
			\hline
			\multicolumn{3}{|l|}{\bf Storm}\\
			\hline
			No fault tolerance & 262.04\,MB/s & 901.47\,MB/s \\
			\hline
			With item backup only & 100.65\,MB/s & 571.31\,MB/s \\
			\hline
			With Kafka only & 85.79\,MB/s & 192.66\,MB/s \\
			\hline
			With Kafka + item backup & 80.14\,MB/s & 174.30\,MB/s \\
			\hline
			\hline
			\multicolumn{3}{|l|}{\bf Spark Streaming}\\
			\hline
			No fault tolerance & 178.61\,MB/s & 754.02\,MB/s \\
			\hline
			With RDD + write-ahead logging & 93.61\,MB/s & 466.07\,MB/s \\
			\hline
			With Kafka only & 81.49\,MB/s & 147.28\,MB/s \\
			\hline
			With Kafka + RDD & 75.09\,MB/s & 140.80\,MB/s \\
			\hline
			\hline
			\multicolumn{3}{|l|}{\bf AF-Stream implementation}\\
			\hline
			No fault tolerance & 3.55\,GB/s & 1.92\,GB/s \\
			\hline
			Mini-batch & 358.16\,MB/s & 380.49\,MB/s \\
			\hline
			Upstream backup & 379.61\,MB/s & 373.89\,MB/s \\
			\hline
			Approximate fault tolerance & 3.48\,GB/s & 1.88\,GB/s \\
			\hline
		\end{tabular}
	}
	\caption{Experiment~1: Comparisons with existing fault tolerance approaches.}
	\label{tab:exp1}
	\vspace{-12pt}
\end{table}

Table~\ref{tab:exp1} shows the throughput of different fault tolerance
approaches.  Both Storm and Spark Streaming see throughput drops when fault
tolerance is used.  Compared to the throughput without fault tolerance, even
the most modest case degrades the throughput by around 37\% (i.e., Storm's
item backup for Grep).  We find that the bottlenecks of both systems are
mainly due to frequent item backups.  Also, Kafka integration achieves even
lower throughput since Kafka incurs extra I/Os to read traces from disk.
In contrast, AF-Stream with approximate fault tolerance issues fewer backup
operations.  It achieves 3.48\,GB/s for Grep and 1.78\,GB/s for WordCount, both
of which are close to when AF-Stream disables fault tolerance.  Note that
AF-Stream outperforms Spark Streaming and Storm even when they disable fault
tolerance.  The reason is that AF-Stream has a more simplified implementation.

\paragraph{Experiment~2 (Impact of thresholds on performance).} We examine how
the thresholds $\Theta$, $L$, and $\Gamma$ affect the performance of
AF-Stream in different aspects.  Since Grep does not keep any state, we focus
on WordCount.  We vary $\Theta$ and $L$, and fix $\Gamma\!=\!10^3$.   When
$\Theta$ and $L$ are sufficiently large (i.e., close to infinity), we in
essence disable both state and item backups, respectively.

Figure~\ref{fig:exp2}(a) presents the throughput of AF-Stream versus $L$ for
different settings of $\Theta$ and the case of disabling state backups.
Compared to disabling state backups, the throughput loss is 33\% when
$\Theta=1$, but we reduce the loss to 10\% by setting $\Theta=10$.

Figure~\ref{fig:exp2}(b) also presents the recovery time when recovering a
worker failure, starting from the time when a new worker process is resumed
until it starts normal processing.  A smaller $\Theta$ implies a longer
recovery time, as we need to make backups for more updated state values in
partial backup (\S\ref{subsubsec:state}).  Nevertheless, the recovery time in
all cases is less than one second.

\begin{figure}[t]
	\centering
	\begin{tabular}{@{\ }c@{\ }c}
		\includegraphics[width=1.65in]{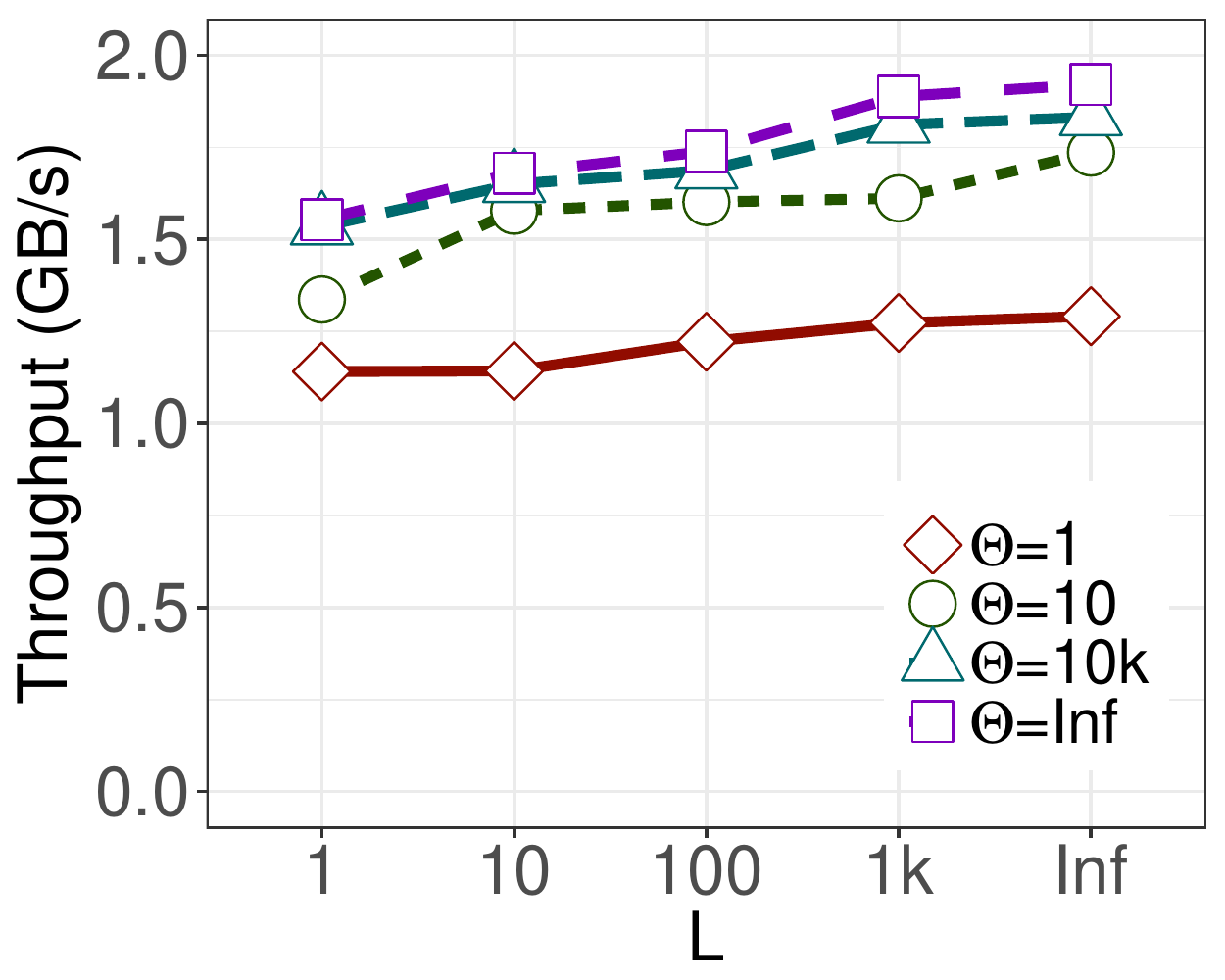} &
		\includegraphics[width=1.65in]{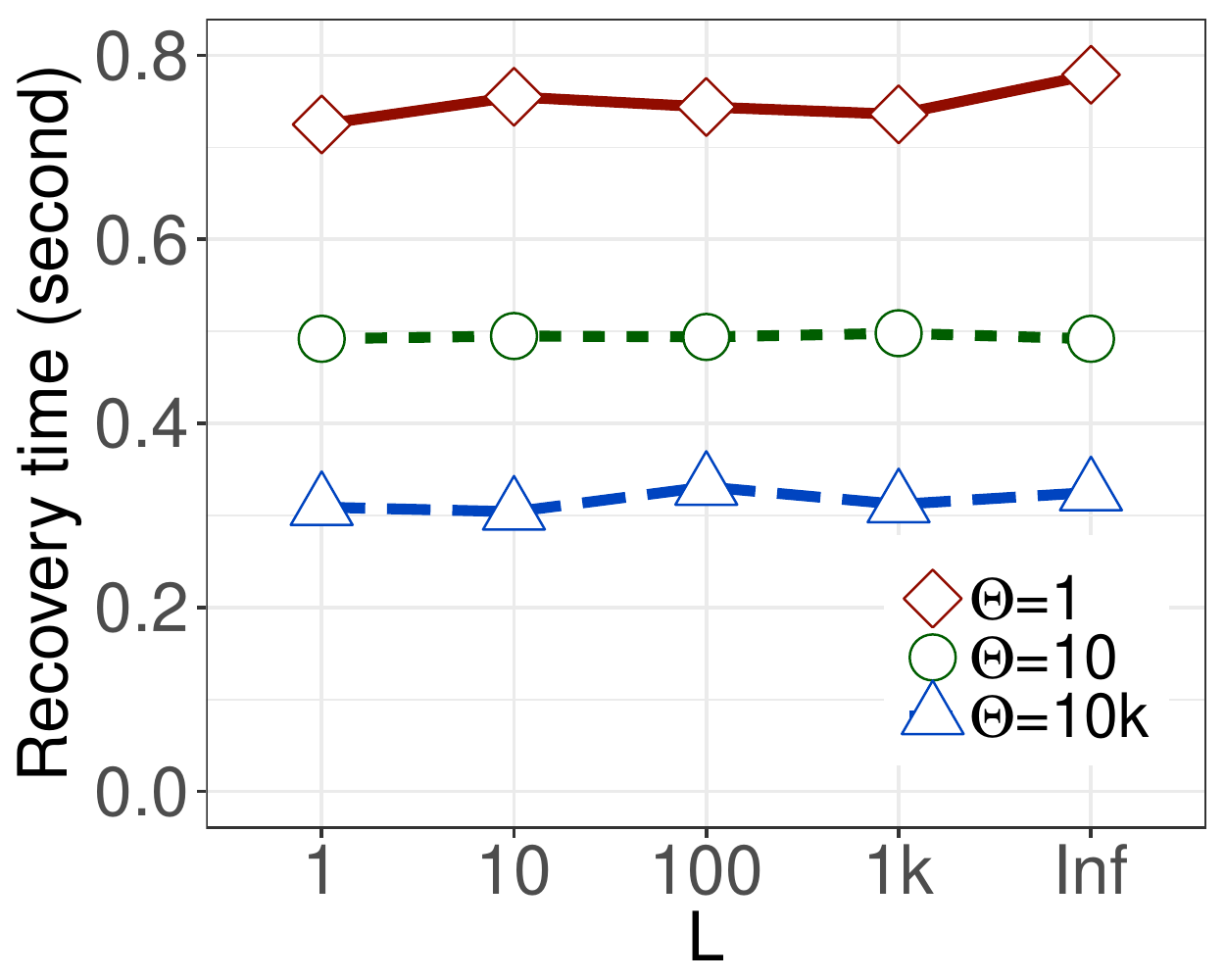}
		\vspace{-3pt}\\
		{\small (a) Throughput} &
		{\small (b) Recovery time}\\
		\includegraphics[width=1.65in]{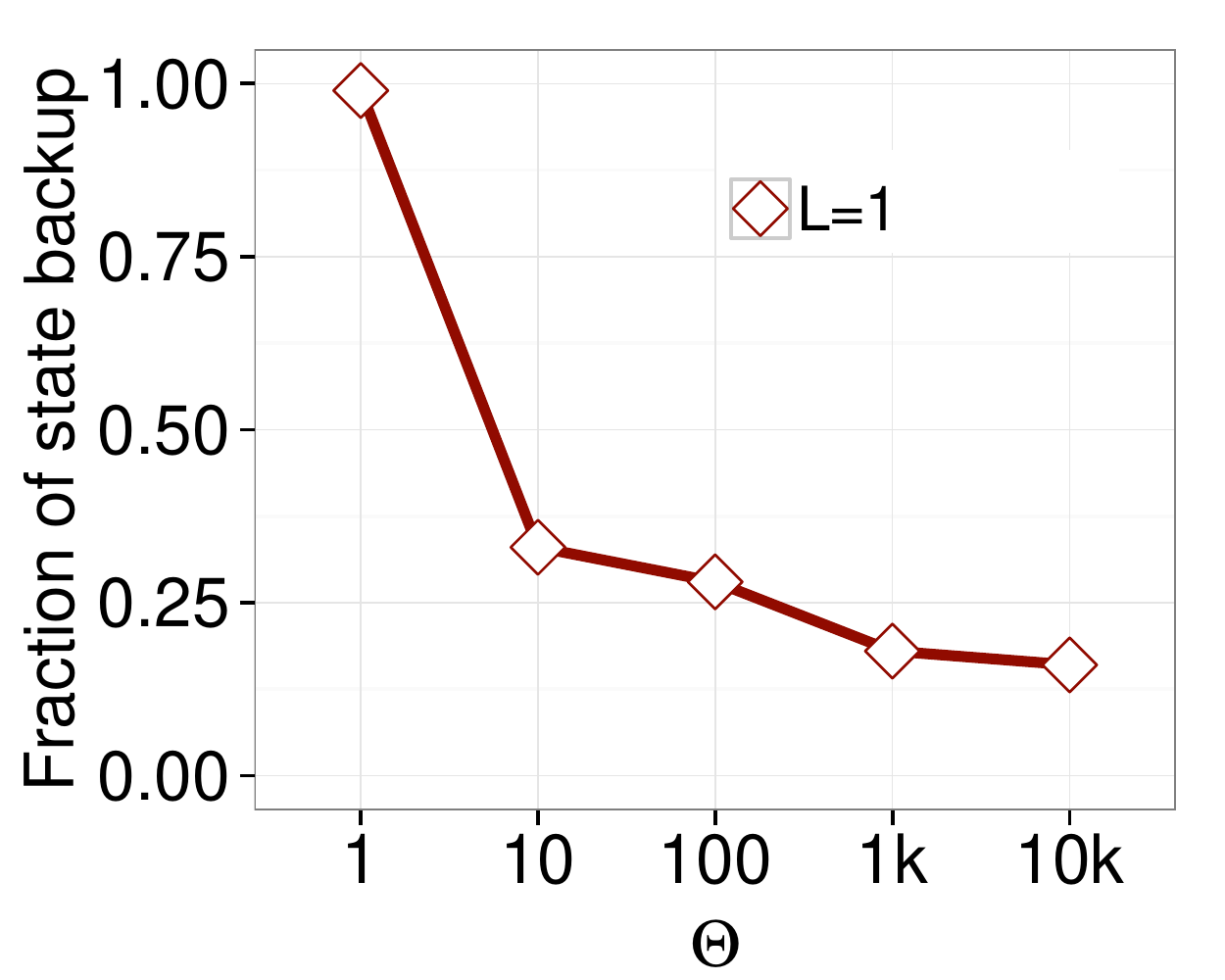} &
		\includegraphics[width=1.65in]{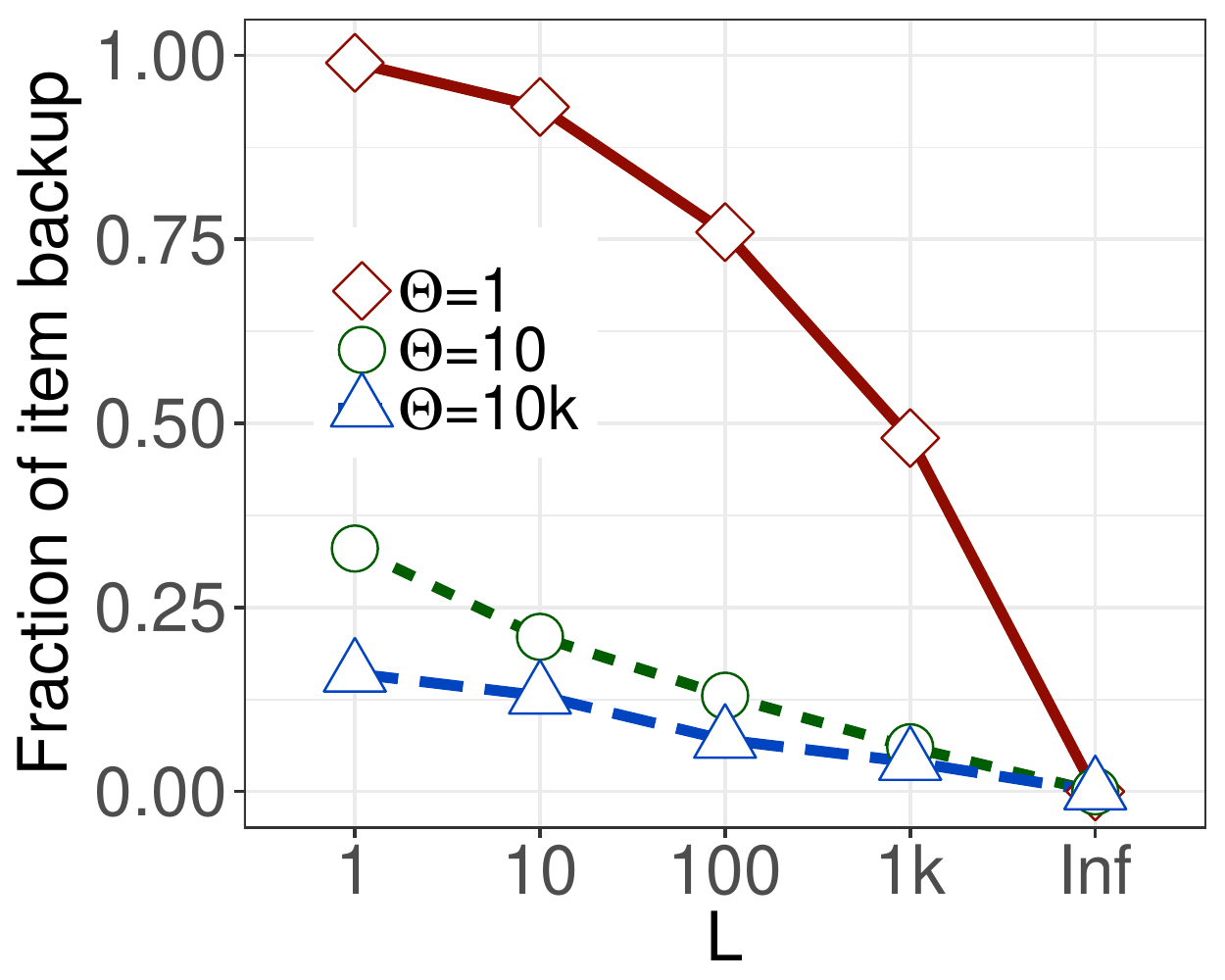}
		\vspace{-3pt}\\
		{\small (c) Fraction of state backups} &
		{\small (d) Fraction of item backups}
	\end{tabular}
	\caption{Experiment~2: Impact of thresholds on WordCount.}
	\label{fig:exp2}
\end{figure}

We justify our throughput results by showing the fractions of state and item
backup operations over the total number of items.  Figure~\ref{fig:exp2}(c)
shows that AF-Stream issues state backups for less than 30\% of items when
$\Theta\ge10$ (where we fix $L\!=\!1$).
Figure~\ref{fig:exp2}(d) shows that increasing $\Theta$ also
reduces the fraction of item backups (e.g., to less than 20\% when
$\Theta\ge10$), mainly because the compute threads have more available
resources to process items rather than perform state backups. This reduces the
number of unprocessed items, thereby issuing fewer item backups.

\subsubsection{Performance-accuracy Trade-offs}
\label{subsubsec:real}

We examine how AF-Stream trades between performance and accuracy.  We evaluate
its throughput when no failure happens, while we evaluate its accuracy after
recovering from system failures in which all workers fail.
To mimic a system failure, we inject a
special item in the stream to a worker.  When the worker reads the special
item, it sends a remote stop signal to kill all worker processes.  We then
resume all worker processes, recover all backup states, and replay the backup
items.  We inject the special item multiple times to generate multiple
failures over the entire stream.  We consider the other three streaming
algorithms, which are more complicated than Grep and WordCount.

In the following, we vary $\Theta$ and $L$, and fix $\Gamma\!=\!10^3$.
Here, $\Gamma$ represents the maximum number of unacknowledged items in
upstream workers.  We observe that the actual number of unacknowledged items
is small and they have limited impact on both performance and accuracy in our
experiments.  Thus, we focus on the physical meanings of $L$ and $\Theta$ in
each algorithm and justify our choices of the two parameters.

\begin{figure}[t]
	\centering
	\begin{tabular}{@{\ }c@{\ }c}
		\includegraphics[width=1.65in]{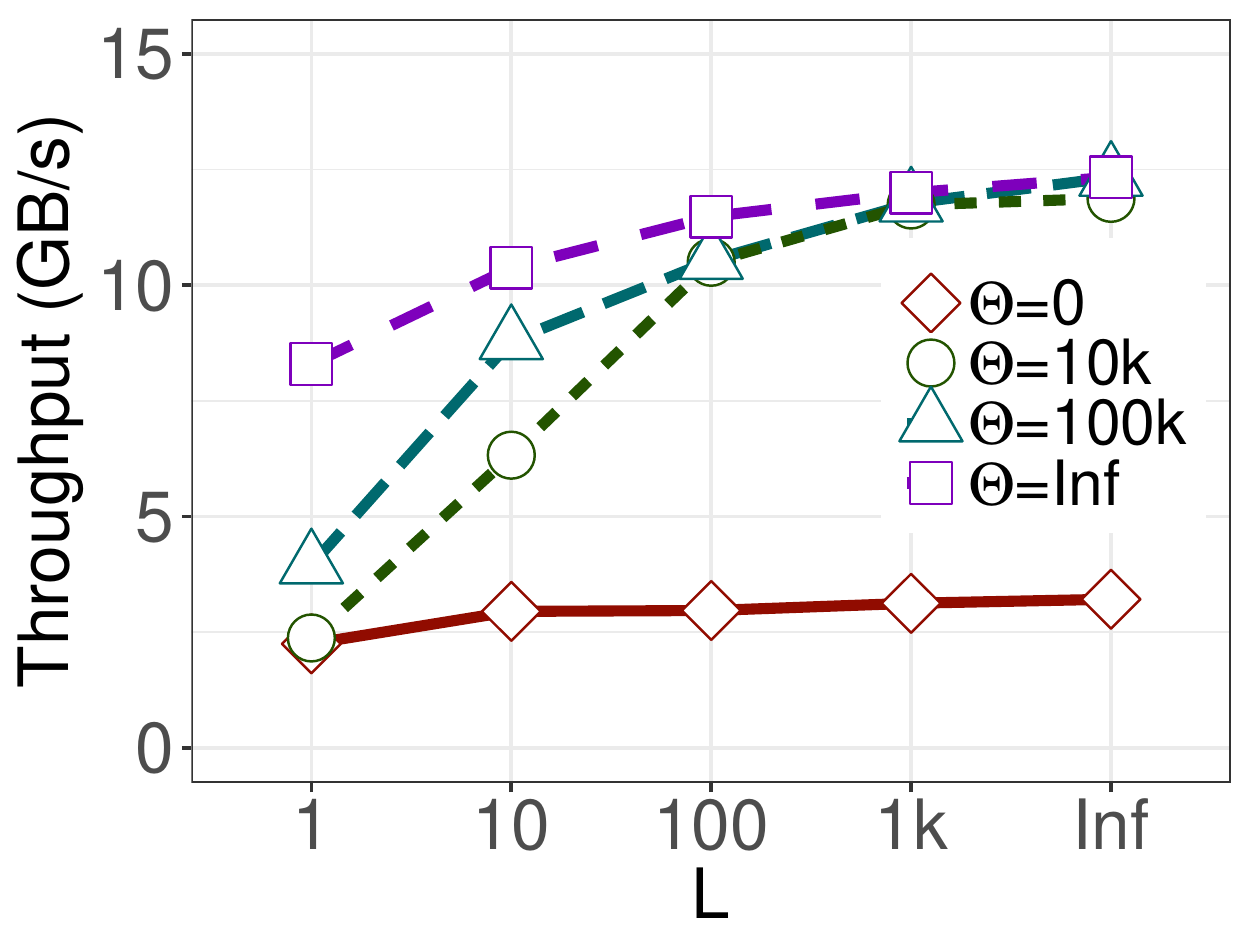} &
		\includegraphics[width=1.65in]{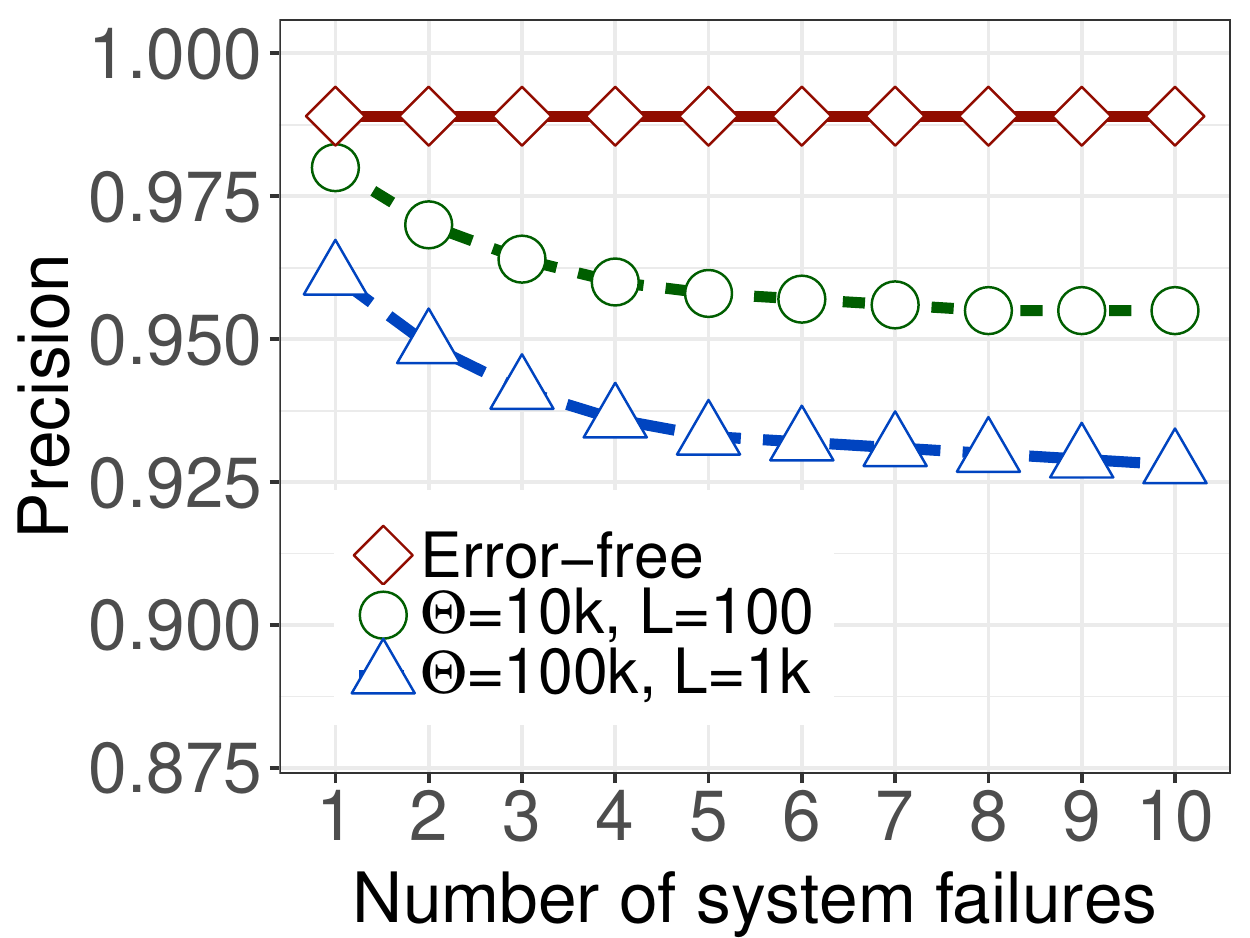} \\
		{\small (a) Throughput} &
		{\small (b) Precision}
	\end{tabular}
	\caption{Experiment~3: Heavy hitter detection.}
	\label{fig:exp3}
\end{figure}

\paragraph{Experiment~3 (Heavy hitter detection).} We perform heavy hitter
(HH) detection in a stream of packet headers from CAIDA \cite{Caida}, where
the total packet header size is 40\,GB.  We define an HH as a source-destination
IP address pair whose total payload size exceeds 10\,MB.  We implement HH
detection based on Count-Min Sketch \cite{Cormode2005a}.  Each worker holds a
Count-Min Sketch with four rows of 8,000 counters each as its state.  Each
packet increments one counter per row by its payload size.  We maintain the
counters in a matrix that has built-in fault tolerance (\S\ref{sec:impl}).
Also, an upstream worker sends local sketches and detection results to a
downstream worker as punctuation items, for which we ensure error-free fault
tolerance (\S\ref{subsubsec:item_backup}).

In our implementation, $\Theta$ is the maximum difference of counter values
(in bytes) between the current state and the most recent backup state and $L$
is the maximum number of unprocessed non-backup packets.
We choose $\Theta\!\le\!10^5$ and $L\!\le\!10^3$.
If the packet size is at most 1,500~bytes, our choices account for at most
1.5\,MB of counter values, much smaller than our selected threshold 10\,MB.

We measure the throughput of AF-Stream as the total packet header size
processed per second.  To measure the accuracy, one important fix is that we
need to address the missing updates due to approximate fault tolerance.  In
particular, the missing updates cause the restored counter values to be
smaller than the original counter values before a failure.  To compensate for
the missing updates, we add each counter by $\Theta/2^k+L\alpha/2^k$ when
restoring counter values after the $k$-th failure (where $k\ge 1$).  This
ensures the zero false negative rate of Count-Min Sketch, while increasing the
false positive rate by a bounded value.  We measure the accuracy by the
precision, defined as the ratio of the number of actual HHs to the number of
returned HHs (which include false positives).

Figure~\ref{fig:exp3} presents the throughput and precision of AF-Stream for
HH detection.  If we disable both state and item backups (i.e., $\Theta$ and
$L$ are close to infinity), the throughput is 12.33\,GB/s, and
the precision is 98.9\% when there is no failure.  With approximate fault
tolerance, if $\Theta\!=\!10^5$ and $L\!=\!10^3$, the throughput
decreases by up to 4.7\%, yet the precision only decreases to 92.8\% after 10
system failures.  If we set $\Theta\!=\!10^4$ and $L\!=\!100$, the
throughput drops by around 15\%, while the precision decreases to 95.5\%.

\paragraph{Experiment~4 (Online join).}  
Online join is a basic operation in stream
database queries.  This experiment considers an online join operation that
correlates two streams of packet headers of different cities obtained from
CAIDA \cite{Caida}, with the total packet header size 40\,GB.
Our goal is to return the tuples of the destination IP address
and timestamp (in units of seconds) that have matching packet headers in both
streams, meaning that both streams visit the same host at about the same time.
We partition the join operator into multiple workers, each of which runs a
Ripple Join algorithm \cite{Haas1999} to sample a subset of packets for
online join at a sampling rate 10\%.  Each worker keeps the sampled packets in
a set with built-in fault tolerance (\S\ref{sec:impl}).

Here, $\Theta$ represents the maximum difference of the numbers of packets
between the current state and the most recent backup state, and $L$ is the
maximum numbers of unprocessed non-backup packets.
We find that our sampling rate can obtain around 100~million packets.  Thus,
we set $\Theta\!\le\!10^5$ and $L\!\le\!10^3$
to account for at most 0.1\% of sampled packets.

We again measure the throughput of AF-Stream as the total packet header size
processed per second.  To measure the accuracy, we issue an aggregation query
for the total number of joined tuples.  Ripple join returns the estimated
number of tuples by dividing the number of joined tuples in the sampled set by
the sampling rate.  We measure the accuracy as the relative estimation error,
defined as the percentage difference of the estimated number from the actual
number without any sampling.

Figure~\ref{fig:exp4} presents the throughput and relative estimation error of
AF-Stream for online join.  If we disable both state and item backups, the
throughput is 6.96\,GB/s.  The relative estimation error is 9.1\% when no
failure happens.  If we enable approximate fault tolerance, the throughput
drops by 5.2\% for $\Theta\!=\!10^5$ and $L\!=\!10^3$, and by 12\% for
$\Theta\!=\!10^4$ and $L\!=\!100$, while the relative estimation error
only increases to 11.3\% and 10.6\%, respectively, even after 10 system
failures.

\begin{figure}[t]
	\centering
	\begin{tabular}{@{\ }c@{\ }c}
		\includegraphics[width=1.65in]{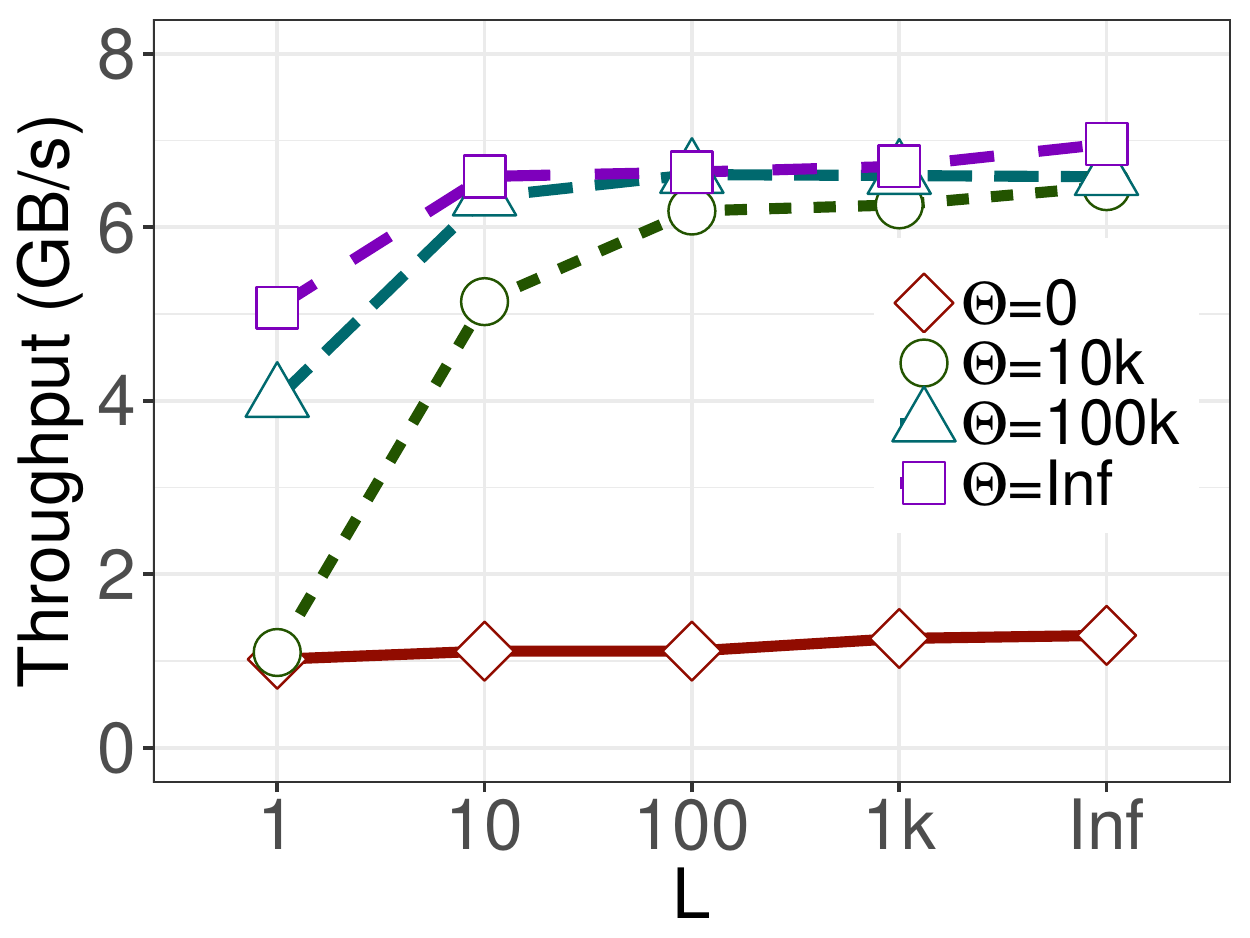} &
		\includegraphics[width=1.65in]{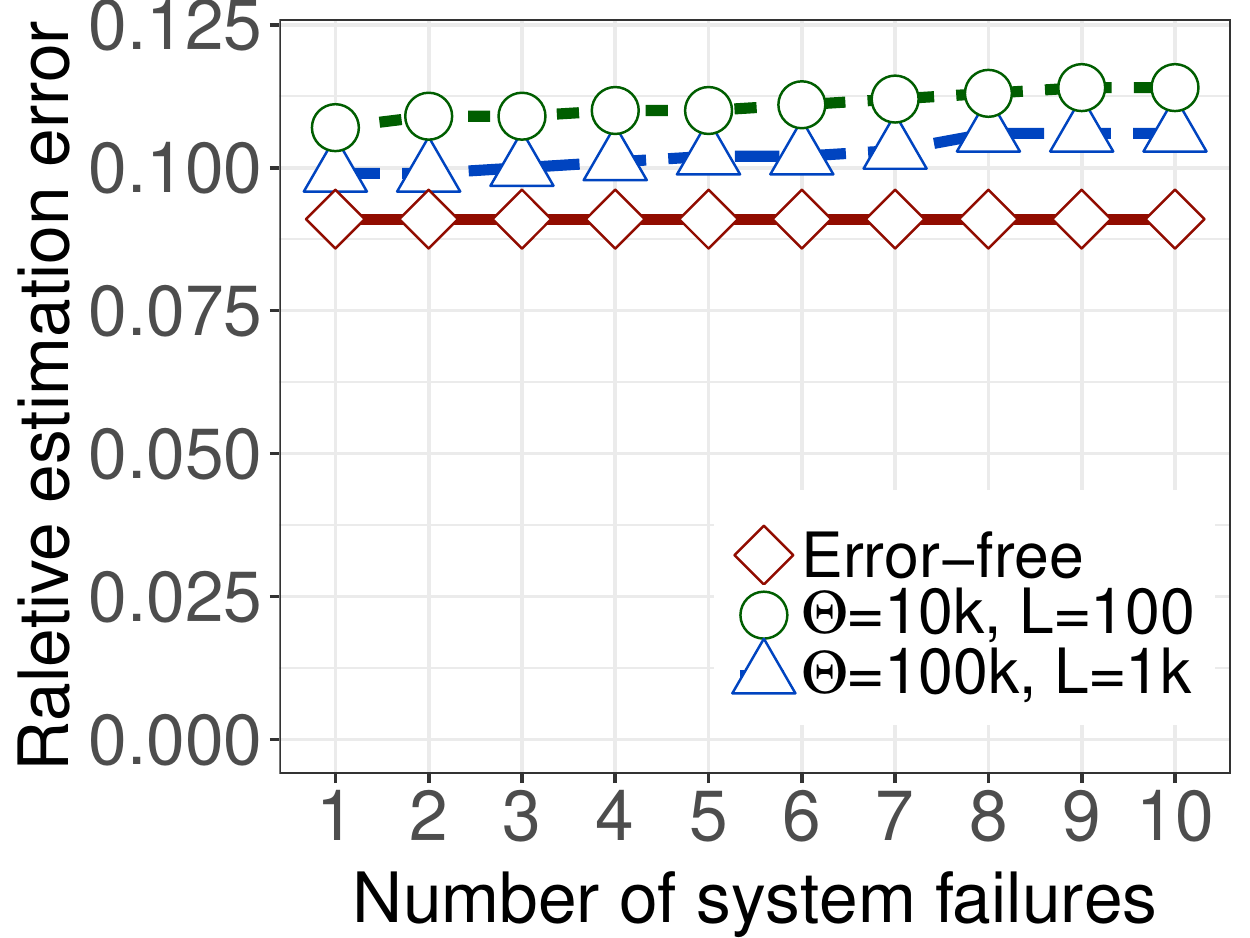} \\
		{\small (a) Throughput} &
		{\small (b) Relative estimate error}
	\end{tabular}
	\caption{Experiment~4: Online join.}
	\label{fig:exp4}
\end{figure}

\paragraph{Experiment~5 (Online logistic regression).} Logistic regression is
a classical algorithm in machine learning.  We use online logistic regression
to predict advertisement click-through rates for a public trace in KDD Cup
2012 \cite{Niu2012}.  We extract a trace with 150~million tuples, each of
which is associated with a label and multiple features.  However, we find that
40~million of them have a missing value in one of the features.  We remove
these 40~million anomalous tuples from the trace and use the remaining
110~million tuples (with a total size 42\,GB) in our evaluation; nevertheless,
we include those anomalous tuples in our convergence evaluation to show the
robustness of AF-Stream in model training (\S\ref{subsec:eval_convergence}).  

We evenly divide the entire trace into two halves, one as a training set and
another as a test set.  We train the model with a distributed SGD technique
\cite{Langford2009}, in which each upstream worker trains its local model with
a subset of the training set, and regularly sends its local model to a single
downstream worker (every $10^3$ tuples in our case) in the form of a
punctuation item (\S\ref{subsubsec:item_backup}).  The online logistic
regression algorithm has a feedback loop, in which the downstream worker
computes the average of the model parameters to form a global model, and sends
the global model to each upstream worker in the form of a feedback item.  The
upstream worker updates its local model accordingly.  Each upstream worker
stores the model parameters in a hash table as its state.  In this experiment,
we configure AF-Stream to run in the ASP consistency model (\S\ref{sec:impl}).

Here, $\Theta$ represents the Euclidean distance of the model parameters, and
$L$ is the maximum number of unprocessed non-backup tuples for model training.
As our model has more than $10^6$ features, we set $\Theta\!\le\!10$ and
$L\!\le\!10^3$ to limit the errors to the model parameters.  This setting
implies an average Euclidean distance of less than
$\sqrt{\Theta/10^6}\!=\!0.01$ for each feature parameter and at most
$L/(110\times10^6)\!=\!0.001\%$ of lost tuples. 

We measure the throughput as the number of tuples processed per second.  To
measure the accuracy, we predict a label for each tuple in the test set based
on its features by logistic regression, and check if the predicted label is
identical to the true label associated with the tuple.   We measure the
accuracy as the prediction rate, defined as the fraction of correctly
predicted tuples over all tuples in the test set.

Figure~\ref{fig:exp5} presents the throughput and prediction rate of AF-Stream
for online logistic regression.  The throughput is 62,000~tuples/s when we
disable state backups.  If we set $\Theta\!=\!10$ and $L\!=\!10^3$, or
$\Theta\!=\!1$ and $L\!=\!100$, the throughput is only less than that without
state and item backups by up to 0.3\%.  The reason is that the model has a
large size with millions of features and incurs intensive computations.  Thus,
the throughput drops due to state backups become insignificant.  In
addition, the prediction rate is 94.3\% when there is no failure.  The
above two settings of $\Theta$ and $L$ decrease the prediction rate to
90.4\% and 92.9\%, respectively, after 10 system failures.

\begin{figure}[t]
	\centering
	\begin{tabular}{@{\ }c@{\ }c}
		\includegraphics[width=1.65in]{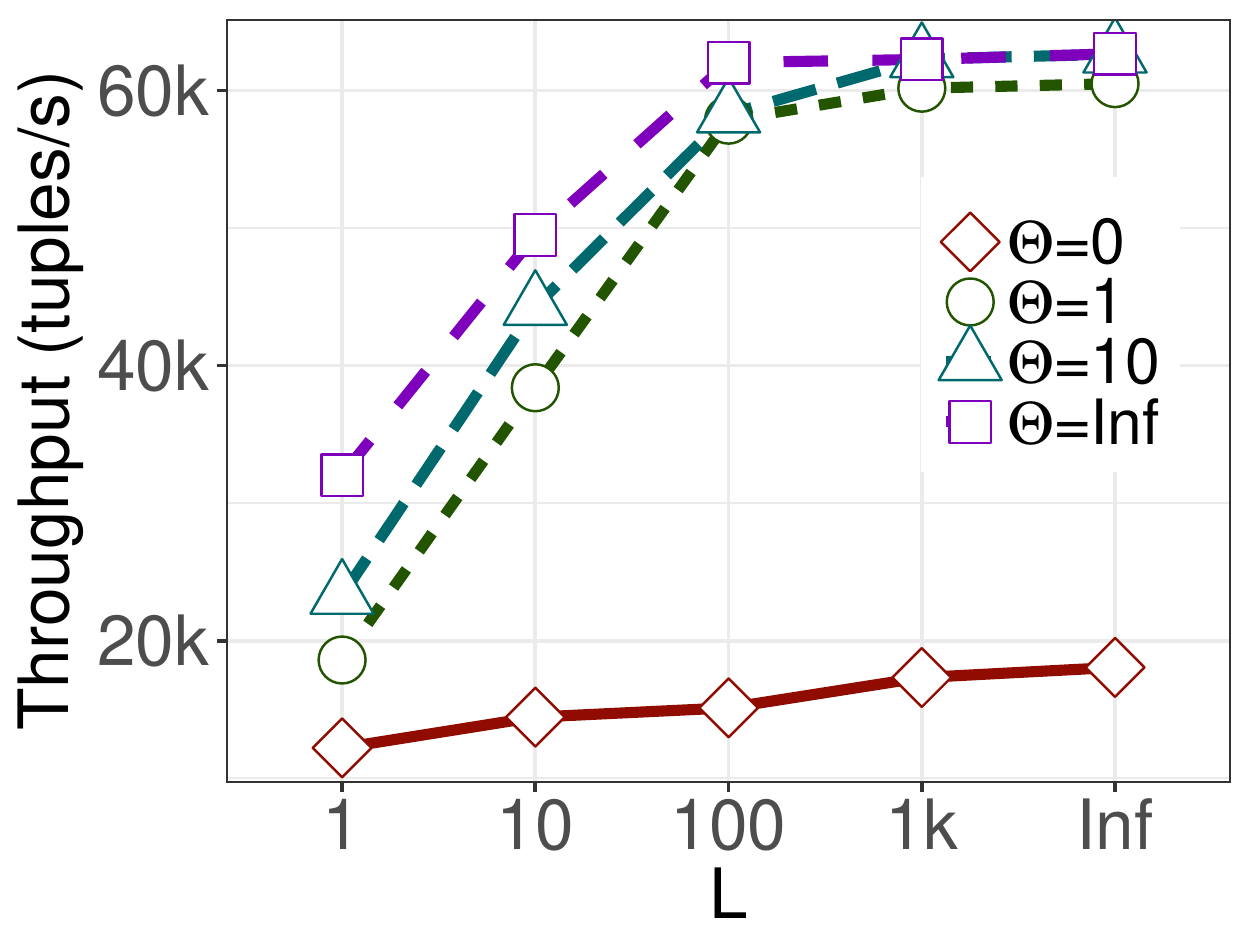} &
		\includegraphics[width=1.65in]{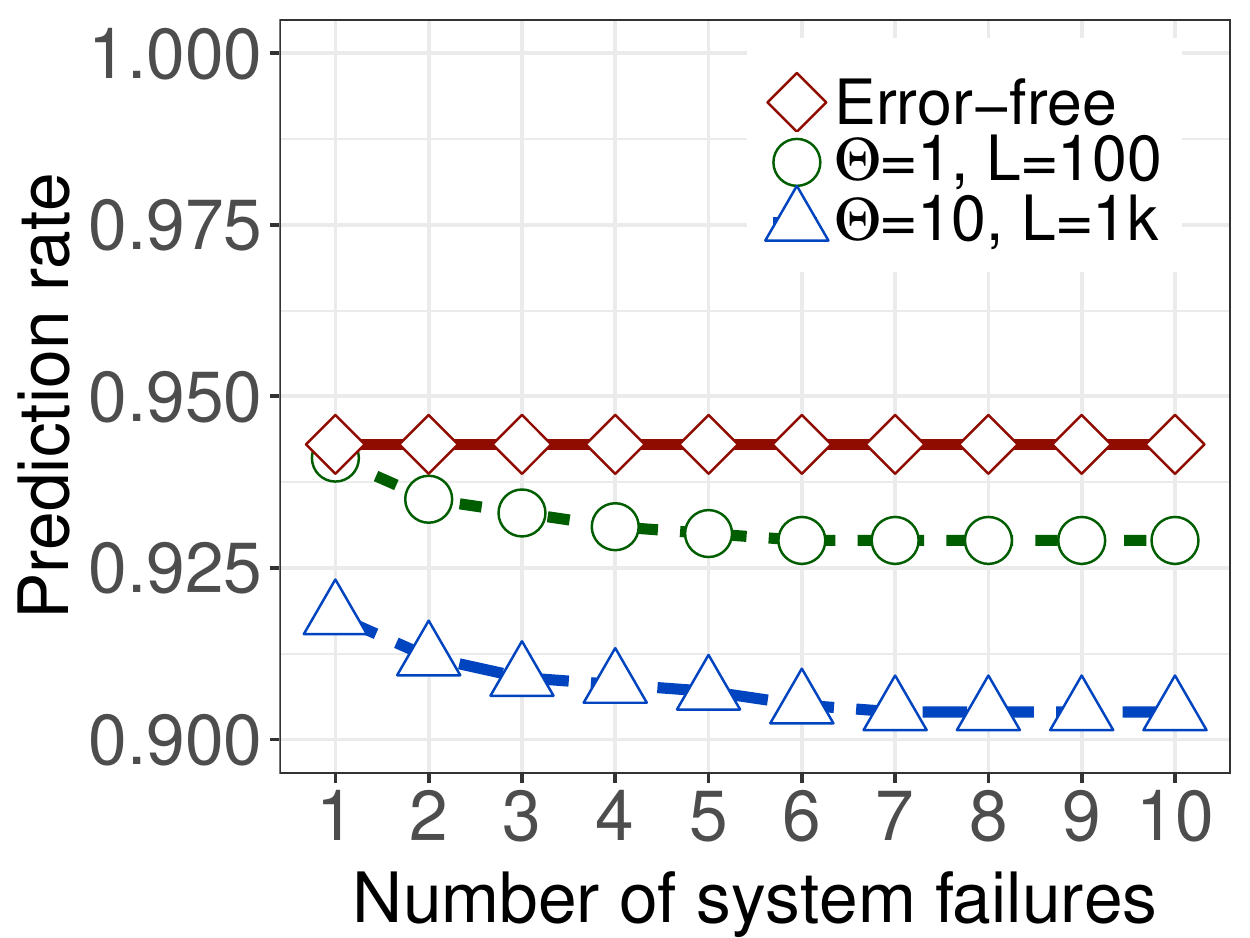} \\
		{\small (a) Throughput} &
		{\small (b) Prediction rate}
	\end{tabular}
	\caption{Experiment~5: Online logistic regression.}
	\label{fig:exp5}
\end{figure}

\begin{table}[!t]
	\centering
	\label{tab:cmp}
	\renewcommand{\arraystretch}{1.0}
	{\small
		\begin{tabular}[c]{|c|c|c|c|}
			\hline
			& HH detection,
			& Online join,
			& Online LR, \\
			& $\Theta=10^5$
			& $\Theta=10^5$
			& $\Theta=10$ \\
			\hline
			\hline
			\makecell{\bf No fault\\\bf tolerance}
			& 12.33 GB/s & 6.96 GB/s & 62,699 tuples/s \\
			\hline
			\makecell{\bf Mini-batch}
			& 0.91 GB/s & 0.89 GB/s & 14,770 tuples/s\\
			\hline
			\makecell{\bf Upstream\\\bf backup}
			& 0.89 GB/s & 0.87 GB/s & 12,827 tuples/s\\
			\hline
			\hline
			\makecell{\bf Approximate\\\bf fault tolerance}
			& 11.77 GB/s & 6.59 GB/s & 62,292 tuples/s \\
			\hline
		\end{tabular}
	}
	\caption{Comparisons of fault tolerance approaches for HH detection, online
		join and online logistic regression (LR) with $L\!=\!10^3$ and
		$\Gamma\!=\!10^3$.}
\end{table}

\paragraph{Discussion.}
Our results demonstrate how AF-Stream addresses the trade-off between
performance and accuracy for different parameter choices.  We also implement
the three algorithms under existing fault tolerant approaches as in
Experiment~1, and Table~\ref{tab:cmp} summarizes the results.  We observe that
both mini-batch and upstream backup approaches reduce the throughput to less
than 25\% compared to disabling fault tolerance, while approximate fault
tolerance achieves over 95\% of the throughput.

\subsection{Convergence Evaluation on a Local Cluster}
\label{subsec:eval_convergence}

\begin{figure}[t]
	\centering
	\begin{tabular}{@{\ }c@{\ }c}
		\includegraphics[width=1.65in]{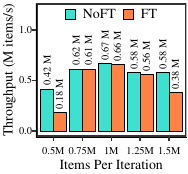} & 
		\includegraphics[width=1.65in]{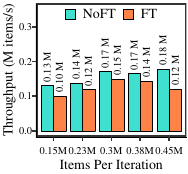}  \\ 
		\mbox{\small (a) KDDCup, Spark Streaming} &
		\mbox{\small (b) HIGGS, Spark Streaming} \\
		\includegraphics[width=1.65in]{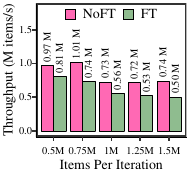} & 
		\includegraphics[width=1.65in]{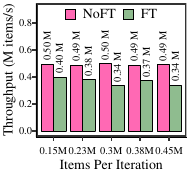} \\
		\mbox{\small (c) KDDCup, Flink} &
		\mbox{\small (d) HIGGS, Flink} \\
		\includegraphics[width=1.65in]{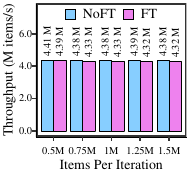} &
		\includegraphics[width=1.65in]{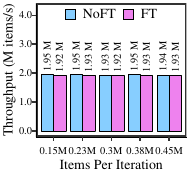} \\
		\mbox{\small (e) KDDCup, AF-Stream} &
		\mbox{\small (f) HIGGS, AF-Stream} 
	\end{tabular}	
	\caption{Experiment~6: Impact of fault tolerance on throughput (NoFT: fault
		tolerance disabled; FT: fault tolerance enabled).}	
	\label{fig:th}
\end{figure}

\begin{figure*}[t]
	\centering
	\begin{tabular}{@{\ }c@{\ }c}
		\includegraphics[width=3.2in]{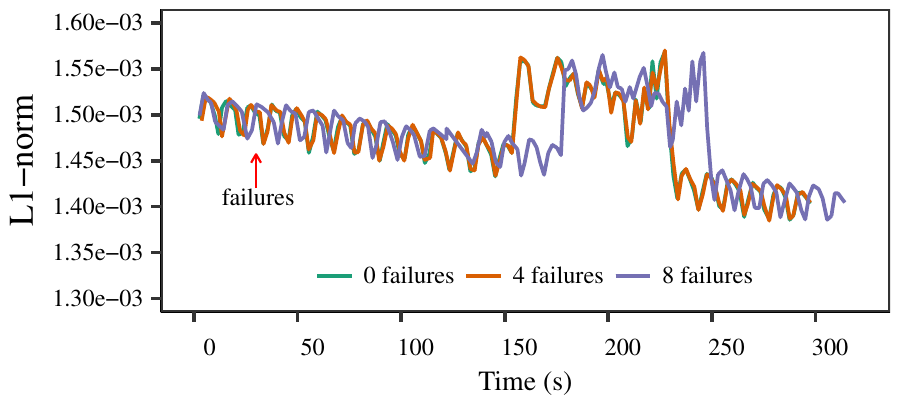} &
		\includegraphics[width=3.2in]{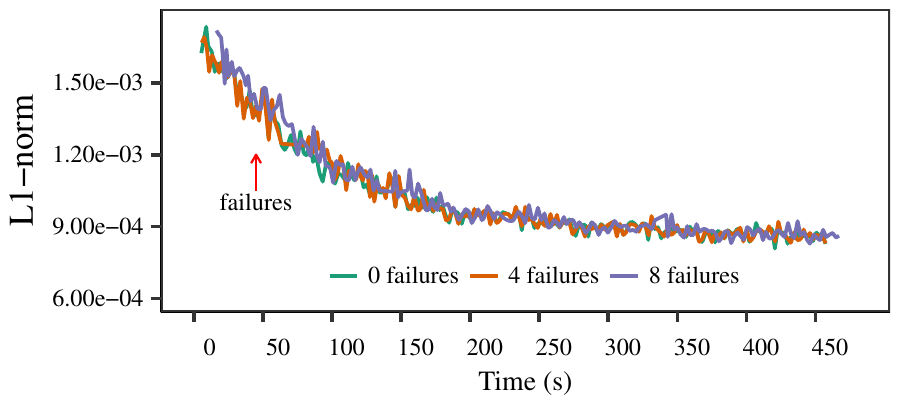}\\
		{\small (a) KDDCup, Spark Streaming} &
		{\small (b) HIGGS, Spark Streaming} 
		\vspace{3pt}\\
		\includegraphics[width=3.2in]{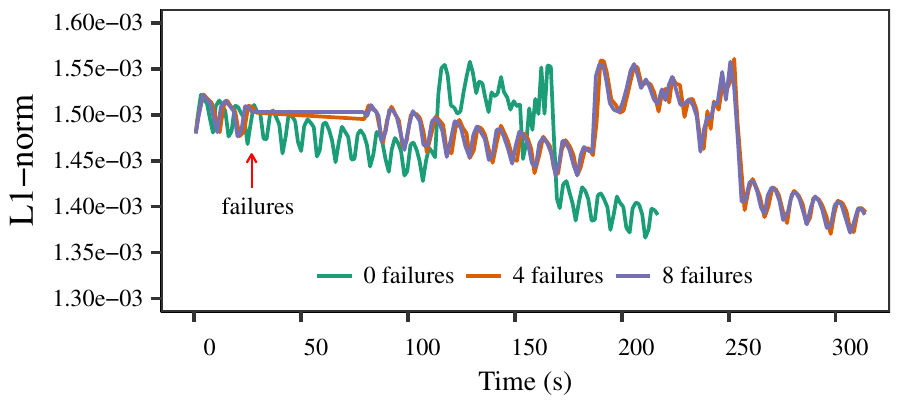} &
		\includegraphics[width=3.2in]{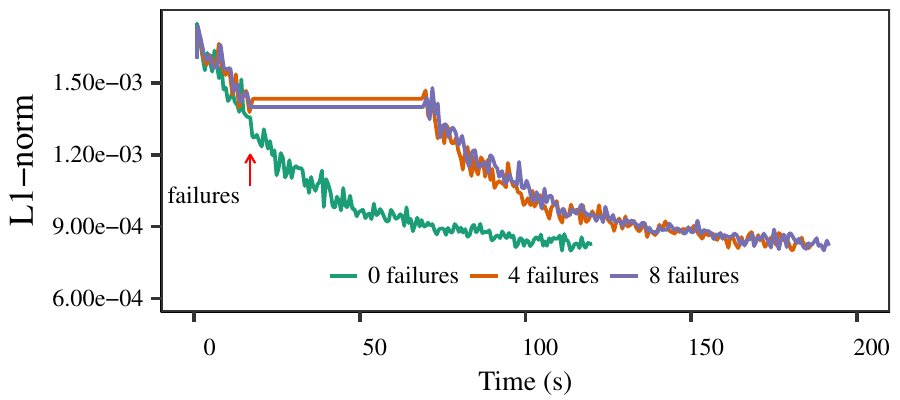}\\
		{\small (c) KDDCup, Flink}  &
		{\small (d) HIGGS, Flink}
		\vspace{3pt}\\
		\includegraphics[width=3.2in]{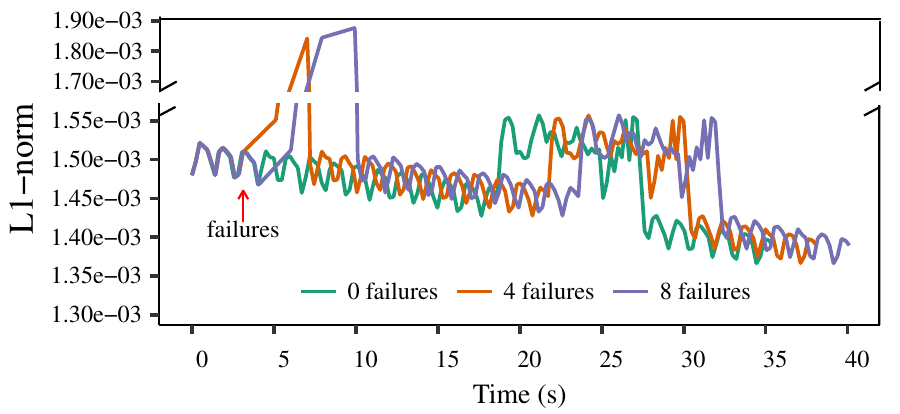} &
		\includegraphics[width=3.2in]{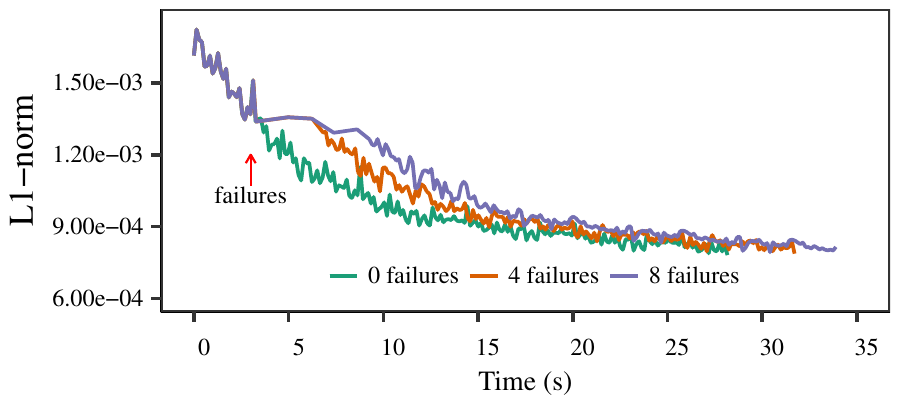}\\
		{\small (e) KDDCup, AF-Stream}  &
		{\small (f) HIGGS, AF-Stream}
	\end{tabular}	
	\caption{Experiment~7: Convergence under failures.}
	\label{fig:failure}
\end{figure*}

We further show how AF-Stream preserves the model convergence of online
learning with respect to failure patterns, traces, and parameter choices.
We conduct our experiments on a local cluster with a total of 12 machines,
each equipped with an Intel Core i5-7500 3.40\,GHz quad-core CPU, 32\,GB RAM,
and a TOSHIBA DT01ACA100 7200\,RPM 1\,TB SATA disk.   All machines run Ubuntu
16.04 LTS and are connected via a 10\,Gb/s switch. 

We focus on the model training of online logistic regression, and compare the
convergence performance under AF-Stream with that under Spark Streaming
(v1.6.3) \cite{Zaharia2013} and Flink (v1.7.2) \cite{Carbone2015}; both Spark
Streaming and Flink can achieve error-free fault tolerance.
Recall that Spark Streaming partitions a stream into {\em mini-batches} on
small time-steps and issues both state backups as checkpoints and item backups
via write-ahead logging, while Flink issues state backups as checkpoints and
persists items in Kafka for item backups (\S\ref{subsec:abstraction}).

We run the model training of online logistic regression on two real-world
traces.  The first trace, denoted by KDDCup, is the same public trace from KDD
Cup 2012 \cite{Niu2012} used in Experiment~5 in \S\ref{subsec:eval_ec2}.  We
use all 150~million tuples in the trace, including the 40~million anomalous
tuples (\S\ref{subsec:eval_ec2}).  We preprocess the trace so that each item
has 11 features.
The second trace, denoted by HIGGS, contains items for the
classification of signal and background processes in high-energy physics
experiments \cite{Baldi2014}.  Note that we concatenate the original HIGGS
trace five times to obtain 55~million items with 28 features each. 

For each trace, AF-Stream, Flink, and Spark Streaming train a regression model
with SGD; in Spark Streaming, we implement SGD using MLlib \cite{Meng2016}. 
In Flink, since its machine learning library (called FlinkML)
only supports batch processing, we re-implement SGD based on the
implementation in MLlib. 
We partition each trace across 11 local workers in 11 different machines,
each of which iteratively performs gradient computations and sends the results
to a global worker in the remaining machine.  We evaluate the performance and
accuracy of our model training via three metrics: throughput, convergence
time, and the L1-norm difference between the models in the last and current
iterations.  For AF-Stream, by default, we set $\Gamma\!=\!1000$, and consider
different values of $L$ and $\Theta$ in our experiments (in 
Appendix~\ref{sec:appendix_convergence}, we show that the effect of $\Gamma$
is included in state divergence, which is controlled by $\Theta$). 

\paragraph{Experiment~6 (Impact of fault tolerance on throughput).} We first
examine how fault tolerance mechanisms influence the throughput of model
training.  For fair comparisons, we configure all systems to operate in the
BSP consistency model (\S\ref{sec:impl}) and process streaming items in
mini-batches, each of which corresponds to an iteration.
For Flink, to realize the BSP model, we use its {\tt{countWindowAll}} API to
aggregate the streaming items in {\em tumbling windows} that resemble
mini-batches.  We set the stream input rates as 0.5~million and 0.15~million
items per second for KDDCup and HIGGS, respectively.  We vary the 
{\em iteration size} (i.e., the number of items being processed per iteration)
in our evaluation.
For AF-Stream, we set $L\!=\!2000$ and vary $\Theta$ as 0.1, 1, and 10 and
obtain the average results. 

Figure~\ref{fig:th} compares the throughput with and without fault tolerance
versus the iteration size. Recall from Experiment~1 that AF-Stream has high
absolute throughput due to a simplified implementation.  Thus,
our focus here is not to compare the absolute throughput of different
systems.  Instead, we focus on comparing the performance differences with and
without enabling fault tolerance.

We first consider Spark Streaming.  For KDDCup, the fault tolerance overhead
in Spark Streaming's throughput
(Figure~\ref{fig:th}(a)) depends on the iteration size.  If the iteration size
is 0.5~million items per iteration, enabling fault tolerance drops the
throughput by 57\%.  The reason is that the number of mini-batches over the
entire stream becomes high and Spark Streaming needs to trigger a high number
of state backups.  As the iteration size increases (i.e., the number of
mini-batches decreases), the fault tolerance overhead is negligible.
However, when the iteration size reaches 1.5~million items per iteration,
enabling fault tolerance drops the throughput by 34\%, since a worker now
needs more resources to process a large number of items in each iteration
in time.  Spark Streaming shows a similar trend in HIGGS as in KDDCup
(Figure~\ref{fig:th}(b)). When the iteration size is 0.15~million items
per iteration, the throughput drops by 23\%.  The throughput drop reduces
to 13\% on average when the iteration size increases from 0.23~million to
0.38~million items per iteration, but increases again to 32\% when the
iteration size reaches 0.45~million items per iteration.
We also observe similar results for Flink, in which enabling
fault tolerance drops the throughput by
25\% and 26\% on average in KDDCup and HIGGS, respectively
(Figures~\ref{fig:th}(c) and \ref{fig:th}(d)).  
On the other hand, AF-Stream's throughput (Figures~\ref{fig:th}(e) and
\ref{fig:th}(f)) remains stable for any iteration size in both traces.  It
drops by only 1.14\% and 1.06\% on average for KDDCup and HIGGS, respectively,
when approximate fault tolerance is used. 

\begin{figure*}[t]
	\centering
	\begin{tabular}{@{\ }c@{\ }c}
		\includegraphics[width=3.2in]{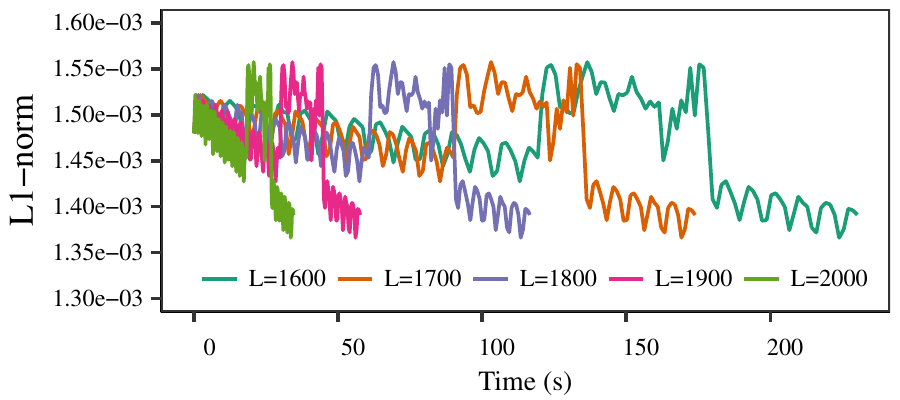} &
		\includegraphics[width=3.2in]{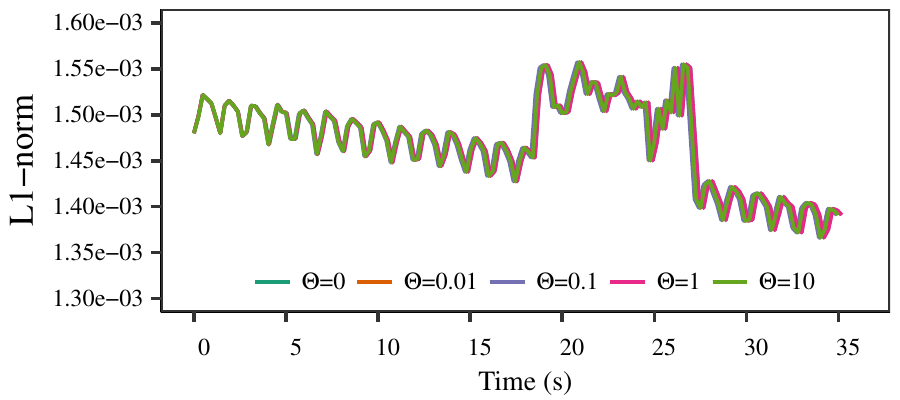}\\
		{\small (a) KDDCup, varying $L$, $\Theta\!=\!0.01$} &
		{\small (b) KDDCup, varying $\Theta$, $L\!=\!2000$}
		\vspace{3pt}\\
		\includegraphics[width=3.2in]{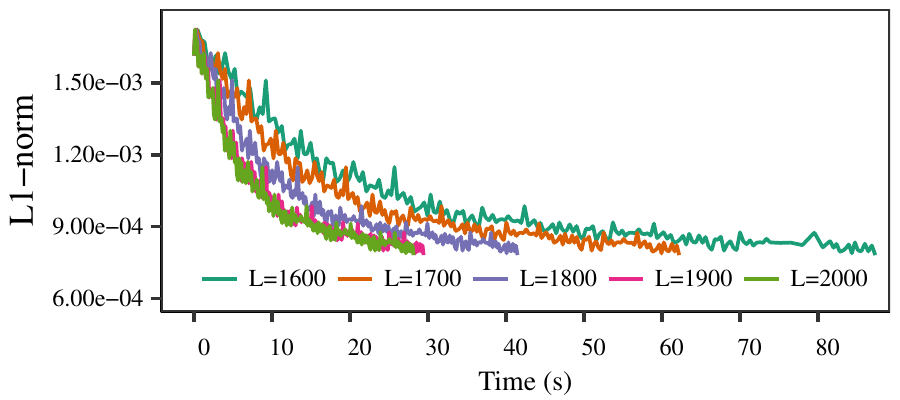} &
		\includegraphics[width=3.2in]{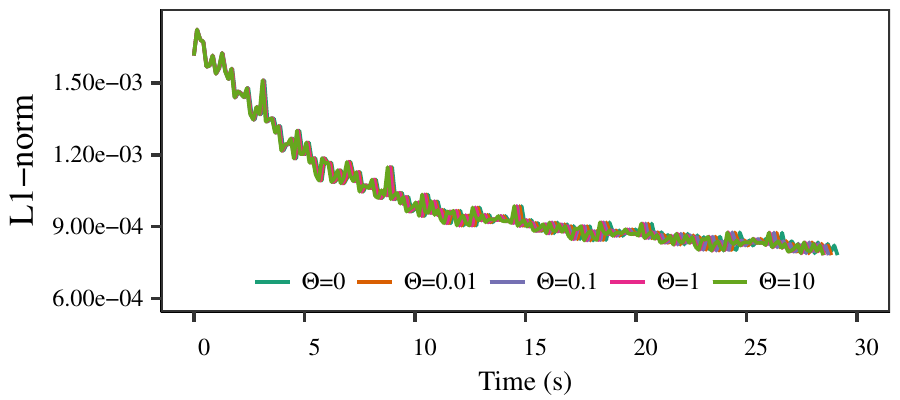}\\
		{\small (c) HIGGS, varying $L$, $\Theta\!=\!0.01$}  &
		{\small (d) HIGGS, varying $\Theta$, $L\!=\!2000$} 
	\end{tabular}	
	\caption{Experiment~8: Impact of parameters on the model convergence of
		AF-Stream.}
	\label{fig:parameter}
\end{figure*}

\paragraph{Experiment~7 (Convergence under failures).}  We next evaluate the
model convergence time when the systems operate under failures.  In all
systems, we kill a number of local workers to mimic simultaneous failures
after they have processed 10\% of all input items.  For Spark Streaming, it
will recover the lost states and items and re-distribute the workloads to the
other non-failed workers.  
For Flink, when some workers fail, it needs to first stop the whole job,
including the processing of other non-failed workers. It then restarts the job
from the latest checkpoint.
For AF-Stream, we immediately restart the workers at the same machines where
they fail.  For AF-Stream, we set $L\!=\!2000$ and $\Theta\!=\!0.01$.
All systems operate under the BSP consistency model.  We set the iteration
size as 1~million and 0.3~million items per iteration for KDDCup and HIGGS,
respectively.  We consider the cases when there is no failure and when we
issue four or eight failures.  Figure~\ref{fig:failure} shows the results. 

First, note that for KDDCup (Figures~\ref{fig:failure}(a),
\ref{fig:failure}(c), and \ref{fig:failure}(e)), there is a sharp increase in
L1-norm even though there is no failure (from time~155s to 229s in Spark
Streaming, from time~114s to 168s in Flink, and from time~18s to 26s in
AF-Stream).  The reason is attributed to the 40~million anomalous tuples with
a missing feature in that period.  Nevertheless, AF-Stream still ensures that
the model converges even under failures as compared to without failures.

Specifically, Spark Streaming provides error-free fault tolerance.  For KDDCup 
(Figure~\ref{fig:failure}(a)), it has the same convergence pattern even with
four failures as in the no-failure case.  Even when there are eight failures,
the convergence time is only slightly delayed, for example, from 297s in the
no-failure case to 314s (i.e., 5.7\% delay) so as to reach the final L1-norm
equal to 1.4e-3.  For HIGGS (Figure~\ref{fig:failure}(b)), Spark Streaming
has the same convergence pattern with both four and eight failures as in the
no-failure case. 

Similarly, Flink sees delayed convergence under failures
(Figures~\ref{fig:failure}(c) and \ref{fig:failure}(d)).  In particular, it
takes more than 50s to restart the whole job when failures happen.  Thus, with
eight failures, the convergence time increases from 218s in the no-failure case
to 313s to reach the final L1-norm 1.4e-3 in KDDCup, and from 120s in the
no-failure case to 191s to reach the final L1-norm 8e-4 in HIGGS.

AF-Stream has a different convergence pattern when failures happen.  For
KDDCup (Figure~\ref{fig:failure}(e)), the L1-norm sharply increases right
after the failures are issued and recovered, since AF-Stream incurs errors in
gradient computation upon the recovery of failures.  Nevertheless, the errors
are amortized after processing six and ten iterations for four and eight
failures, respectively.  The model converges slightly later than the
no-failure case, as AF-Stream issues more backups before the errors are
amortized.  For example, in the no-failure case, the model takes 35.4s to
converge to the final L1-norm equal to 1.4e-3; where there are four and eight
failures, the convergence time slightly increases to 38.1s (i.e., 7.6\% delay)
and 40.1s (i.e., 13.3\% delay), respectively.  For HIGGS 
(Figure~\ref{fig:failure}(f)), the L1-norm does not significantly change right
after the failures are issued and recovered, yet these failures delay the
convergence as they trigger more backups in AF-Stream as in KDDCup. 
With four and eight failures, the convergence time slightly increases from
28.2s to 31.7s (i.e., 12.4\% delay) and 33.9s (i.e., 20\% delay),
respectively, to reach the final L1-norm 8e-4. 

\begin{figure*}[t]
	\centering
	\begin{tabular}{@{\ }c@{\ }c}
		\includegraphics[width=3.2in]{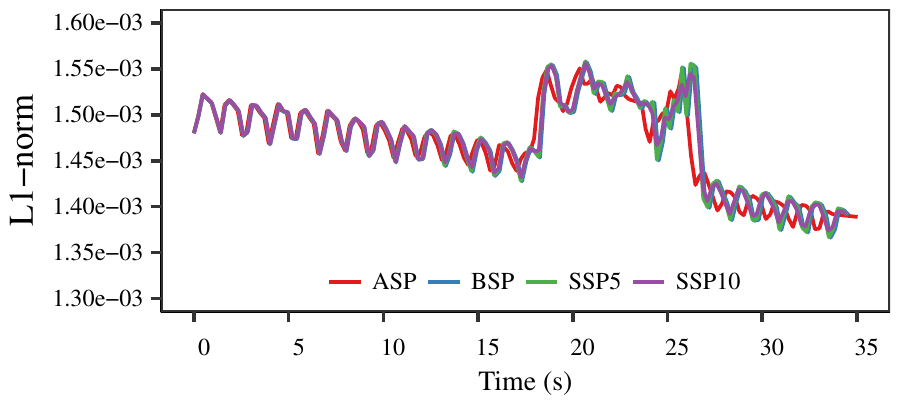} &
		\includegraphics[width=3.2in]{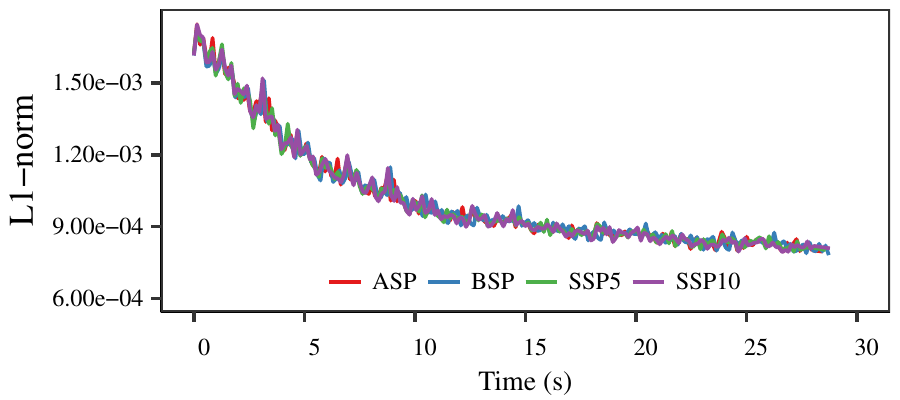}\\
		{\small (a) KDDCup} &
		{\small (b) HIGGS} 
	\end{tabular}
	\caption{Experiment~9: Convergence under different consistency models.}	
	\label{fig:mth}
\end{figure*}

\paragraph{Experiment~8 (Impact of parameters on the model convergence of
	AF-Stream).} We evaluate how various parameters affect the model convergence
of AF-Stream.  We focus on studying $\Theta$ and $L$, as they
affect the convergence performance as shown in the proof of
Theorem~\ref{theorem:sgd} (Appendix~\ref{sec:appendix_convergence}). 
We set the iteration size as 1~million and 0.3~million items per second for
KDDCup and HIGGS, respectively, and run AF-Stream under the BSP consistency
model.  

We first study KDDCup.  Figure~\ref{fig:parameter}(a) shows the model
convergence of AF-Stream for different values of $L$ for KDDCup, where we fix
$\Theta\!=\!0.01$.  If we use smaller $L$, the model convergence time
significantly increases.  The reason is that AF-Stream now issues more item
backups to satisfy the fault tolerance requirement, and such backup overhead
significantly degrades the overall performance.   Note that increasing $L$
beyond 2000 has roughly the same performance as in $L\!=\!2000$, so we only
show the results of $L$ up to 2000 for brevity.  Figure~\ref{fig:parameter}(b)
shows the model convergence of AF-Stream for different values of $\Theta$,
where we fix $L\!=\!2000$.  The model convergence time remains fairly the
same.  It shows that the impact of $L$ is much greater than that of $\Theta$.
We also show the impact of parameters for HIGGS
(Figures~\ref{fig:parameter}(c) and \ref{fig:parameter}(d)), and the
observations are similar.  

\paragraph{Experiment~9 (Consistency models).}
We further study how AF-Stream performs under different consistency models
(i.e., ASP, BSP, and SSP); for SSP, we configure each local worker to use the
stale model for at most five and ten additional iterations (denoted by SSP-5
and SSP-10, respectively).  We set the iteration size as 1 million and 0.3
million items for KDDCup and HIGGS, respectively. We also set $L\!=\!2000$ and
$\Theta\!=\!0.01$.

Figure~\ref{fig:mth} shows the results.  We find that the model convergence
performance is highly similar across different consistency models in both
traces.  
In all cases, the models converge to the same final L1-norm.  This
demonstrates that AF-Stream works well with different consistency models. In
particular, as all local workers have similar processing speeds, we do not see
significant performance differences between BSP and ASP, even though all local
workers need to wait until receiving the feedback from the global worker in
each iteration under BSP (\S\ref{sec:impl}).

\paragraph{Discussion.}
Our results show that AF-Stream preserves the iterative-convergent and
error-tolerant nature of online learning based on stochastic gradient descent,
which also validates the conclusion of our theoretical analysis. In
particular, AF-Stream maintains the model convergence of online logistic
regression in the presence of failures, and the convergence still holds for
different traces and parameter settings.

\section{Related Work}
\label{sec:related}

\noindent
\textbf{Best-effort fault tolerance.}
Some stream processing systems only provide best-effort fault tolerance.
They either discard missing items (e.g., \cite{Neumeyer2010}) or
missing states (e.g., \cite{ApacheFlume, Storm, Kulkarni2015}), or
simply monitor data loss and return the data completeness to developers (e.g.,
\cite{Jain2008, Logothetis2011}).
These systems may have unbounded errors in the face of failures, making the
processing results useless.  In contrast, AF-Stream bounds errors upon
failures with theoretical guarantees.

\paragraph{Error-free fault tolerance.}
Early stream processing systems extend traditional relational databases to
distributed stream databases to support SQL-like queries on continuous
data streams.  They support error-free fault tolerance,
and often adopt {\em active standby} (e.g.,
\cite{Shah2004,Balazinska2005}) or {\em passive standby} (e.g.,
\cite{Cherniack2003,Hwang2005,Hwang2007,Krishnamurthy2010}).
Active standby employs backup nodes to execute the same computations as
primary nodes, while passive standby makes periodic backups for states and
items to backup nodes.  Both approaches, however, are expensive due to
maintaining redundant resources (for active standby) or issuing frequent
backup operations (for passive standby).

Some stream processing systems realize {\em upstream backup}, in which an
upstream worker keeps the items that are being processed in its downstream
workers until all downstream workers acknowledge the completion of the
processing. 
The upstream worker replays the kept items when failures happen.  This
approach is used extensively by previous studies
\cite{Akidau2013,Qian2013,Fernandez2013,Wang2012,Murray2013,Hu2014,Liu2014,Trident}.
Upstream backup generally incurs significant overhead to normal processing,
as a system needs to save a large number of items for possible replays 
(\S\ref{subsubsec:item_backup}).  

To mitigate the impact on normal processing, asynchronous state checkpoints
\cite{Martin2011, Fernandez2014} allow normal processing to be performed in
parallel with state backups.  However, they are not designed for stream
processing and do not address item backups, which could be expensive.
StreamScope \cite{Lin2016} provides an abstraction to handle failure recovery,
but does not address how to trade between performance and accuracy in fault
tolerance as in AF-Stream.

\paragraph{Incremental processing.}
Incremental processing systems extend batch processing systems for stream
processing, by incrementally batching processing at small timescales.
Some systems extend MapReduce \cite{Dean2004} by pipelining mapper and reducer tasks
(e.g., \cite{Condie2010}), while others explicitly divide a stream into mini-batches and run batch-based processing for
each mini-batch (e.g., \cite{Logothetis2010,Chen2010,Bhatotia2011,Zaharia2013}).
Incremental processing systems inherently support error-free fault tolerance
because all data is available to regenerate states upon failures, but they
incur high I/O overhead in saving all items for availability.

\paragraph{Approximation techniques.}
Recent distributed systems have extensively adopted approximation techniques
to improve performance. For example,   BlinkDB \cite{Agarwal2013} is a
database which samples a subset of data to reduce query latency.  Some machine
learning systems (e.g., \cite{Li2014, Wei2015, Xing2015}) defer
synchronization to reduce communication costs.  In the context of stream
processing, JetStream \cite{Rabkin2014} deploys tunable operators that
automatically trade accuracy for bandwidth saving in wide-area stream
processing.  Approximate Spark Streaming \cite{Agarwal2015} samples a subset
of items in continuous streams for actual processing.  AF-Stream differs from
them by leveraging approximations in fault tolerance for distributed stream
processing.

Load shedding \cite{Tatbul2003,Tatbul2007} is a technique that enables an
overloaded stream processing system to drop unprocessed streaming items, with
a trade-off of producing approximate outputs.  Both load shedding and
AF-Stream target different reliability aspects in stream processing: (i) load
shedding addresses system overloading, while AF-Stream focuses on fault
tolerance; (ii) load shedding is activated upon system overloading, while
AF-Stream proactively mitigates backup overhead and provides error bounds in
failure recovery; and (iii) load shedding does not mitigate overhead due to
state backup.  Nevertheless, AF-Stream can also be integrated with load
shedding to address overloading conditions.

LAAR \cite{Bellavista2014} supports adaptive replication by dynamically
activating or deactivating replication for processing elements based on the
service-level agreement requirements, so as to improve processing
capacities.  AF-Stream targets the adaptive backups for internal states and
streaming items, and we provide error bounds independent of the number of
failures. 

Zorro \cite{Pundir2015} introduces approximation techniques to handle failures
in graph processing.  It exploits vertex replication in distributed graph
processing to reconstruct lost states with a high probability. While the idea
is similar to approximate fault tolerance, it is not applicable for stream
processing since streaming items do not exhibit such replication nature.

\paragraph{Convergence of online learning.} The convergence property of online
learning has been well studied and analyzed through game theory
\cite{Hazan2007} or stochastic programming \cite{Bottou2007}.  Many follow-up
studies extend the convergence analysis of online learning for different
settings, such as signal processing paradigms \cite{Busuttil2007}, delayed
updates in parallel computing architectures \cite{Langford2009}, and platforms
with adversarial delays \cite{Quanrud2015}.  In contrast to the current
studies, we analyze the convergence of online learning in distributed stream
processing under approximate fault tolerance. 

\section{Conclusions}
\label{sec:conclusion}

We propose AF-Stream, a distributed stream processing system that
realizes approximate fault tolerance for both internal states and unprocessed
items.  AF-Stream achieves not only high performance by reducing the number of
backup operations, but also high accuracy by bounding errors upon failures
with theoretical guarantees.  It provides an extensible programming model and
exports user-configurable threshold parameters for configuring the
performance-accuracy trade-off.  In addition, it provably preserves the
iterative-convergent and error-tolerant nature of online learning based on
SGD.  Experiments on our AF-Stream prototype demonstrate its high performance
and high accuracy in various streaming algorithms, as well as its model
convergence in various settings.


{\small
\bibliographystyle{abbrv}      
\bibliography{paper}   
}

\begin{appendices}

\section{Error Analysis}
\label{sec:appendix_error}

In this section, we extend the analysis in \S\ref{subsec:ft_analysis} to
derive the error bounds of two streaming algorithms with respect to $\Theta$,
$L$, and $\Gamma$.  

\subsection{Heavy Hitter Detection}

We consider the heavy hitter (HH) detection with Count-Min Sketch
\cite{Cormode2005a}. 
Count-Min Sketch maintains a matrix of counters with $r$ rows and $w$ counters
per row to keep track of the values of the keys.  Given a threshold $\phi$, a
key $x$ is reported as an HH if its estimated value exceeds $\phi$.
Lemma~\ref{lemma:sketch} presents the error probability of HH detection for an
arbitrary key $x$ with the true value $T(x)$.
\begin{lemma}[\cite{Cormode2005a}]
	\label{lemma:sketch}
	Consider a Count-Min Sketch with $r=\log_{1/2}{\delta}$ rows and
	$w=\frac{U}{\epsilon\phi}$ counters per row, where $U$ denotes the sum of the
	values of all keys, and $\epsilon$ and $\phi$ are error parameters.  It
	reports every HH $x$ where $T(x)\ge\phi$ without any error.  For a non-HH $x$
	with $T(x) < (1-\epsilon)\phi$, it is reported as an HH with an error
	probability at most $\delta$. 
\end{lemma}

\begin{theorem}
	\label{theorem:sketch}
	Suppose that AF-Stream deploys a Count-Min Sketch with the same setting as
	Lemma~\ref{lemma:sketch} and the compensation method in Experiment~3
	(\S\ref{subsec:eval_ec2}).  The sketch reports all HHs. On the other hand,
	upon a single failure, for an arbitrary key $x$ with the true value $T(x) <
	(1-\epsilon)\phi-\Theta-\alpha L$, it is reported as an HH with an error
	probability at most $\delta$.
\end{theorem}
\begin{proof}
	With the compensation method, AF-Stream overestimates a key by at most 
	$\Theta + \alpha L$, but does not underestimate any keys.  Thus, we ensure
	that every HH $x$ with $T(x)\ge\phi$.  On the other hand, let $\NT(x)$ denote
	the estimated value of key $x$.  For a non-HH $x$ with $T(x) <
	(1-\epsilon)\phi-\Theta-\alpha L$, it is reported as an HH if and only if its
	estimate value $\NT(x)\ge\phi$.  Thus, 
	$\Pr\{\NT(x)\ge\phi\} = Pr\{\NT(x)-T(x)\ge\phi-T(x)\} \le 
	\Pr\{\NT(x)-T(x)\ge \epsilon\phi+\Theta+\alpha L\}\le\delta$, due to Markov's
	inequality \cite{Cormode2005a}. 
\end{proof}

\paragraph{Example.} Consider a network traffic stream where the total
volume of flows in a window is $U=$~40~GB and flows exceeding $\phi=$~10MB are
HHs.  Let $\delta=~1/16$ and $\epsilon=~1/2$. By setting
$r=\log_{1/2}{\delta}=4$ and $w=\frac{U}{\epsilon\phi}=8192$, a Count-Min
Sketch reports all HHs.  For any non-HH less than 5~MB (i.e.,
$(1-\epsilon)\phi$), it is falsely reported as an HH with an error probability
at most $\delta=1/16$ based on Lemma~\ref{lemma:sketch}.  

Suppose that we run HH detection with Count-Min Sketch on AF-Stream.  Here, we
assume $\alpha=$~1,500, denoting the maximum packet size in bytes.  If we
configure $L=10^3$ and $\Theta=10^5$, then the keys with sizes less than
3.5~MB (i.e., $(1-\epsilon)\phi-L\alpha-\Theta$) will be falsely reported as
HHs with an error probability at most $\delta=1/16$ based on
Theorem~\ref{theorem:sketch}.

\subsection{Stream Database Queries}

We study Ripple Join \cite{Haas1999}, an online join algorithm that samples a
subset of items in two streams and performs join operations on the sampled
subset.  Let $n$ be the number of sampled items.  For aggregation queries
(e.g., SUM, COUNT, and AVG), we denote the true value and estimate value by
$\mu$ and $\hat{\mu}$, respectively.  Ripple Join provides the following
guarantee of the aggregation error. 
\begin{lemma}[\cite{Haas1999}]
	\label{lemma:join}
	When Ripple Join is applied to $n$ sampled items, it guarantees that
	$\Pr\{\hat{\mu}\in[\mu-\epsilon_n, \mu+\epsilon_n]\}\ge 1-\delta$, where
	$\epsilon_n=\frac{z}{\sqrt{n}}$ and $z$ is a constant number depending on
	$\delta$ and the specific aggregation.
\end{lemma}

AF-Stream loses sampled items in failure recovery, which is equivalent to
decreasing the sampling rate in Ripple Join.  Therefore, AF-Stream (slightly)
increases the aggregation error as fewer items are sampled.
Theorem~\ref{theorem:join} quantifies the new aggregation error. 

\begin{theorem}
	\label{theorem:join}
	When AF-Stream applies Ripple Join to $n$ sampled items, it guarantees that 
	$\Pr\{\hat{\mu} \in [\mu-\epsilon_n, \mu+\epsilon_n]\} \ge 1-\delta$, where
	$\epsilon_n=\frac{z}{\sqrt{n-\Theta-L\beta-\Gamma}}$ and $z$ is a constant
	depends on $\delta, n$ and the specific aggregation.
\end{theorem}
\begin{proof}  
	Recall that AF-Stream defines the state divergence as the difference of the
	numbers of items between the current state and the most recent backup state in
	Experiment~4.  Thus, $\alpha = 1$, since each lost item implies a difference
	of one in the number of items.  AF-Stream ensures that the total number of
	lost sampled items is at most $\Theta+\alpha(\Gamma+L\beta) = \Theta + \Gamma
	+ L\beta$.  By replacing $n$ in Lemma~\ref{lemma:join} with
	$n-\Theta-\Gamma-L\beta$, the theorem concludes. 
\end{proof}

\paragraph{Example.}
Consider two streams with $10^9$ items in a window and Ripple Join employs a
sampling rate 10\%, leading to an average of $10^8$ sampled items.  In this
case, the error is less than $\frac{z}{10,000}$ with a probability at least
$1-\phi$ based on Lemma~\ref{lemma:join}.  

Suppose that we run Ripple Join on AF-Stream.  Here, $\beta$ is the maximum
number of items with which an item can join.  Its value depends on how fast a
stream is generated and the number of items that we apply the join operation.
For example, in Experiment~4 (\S\ref{subsec:eval_ec2}), we find that the
number of items being joined is no more than 100, so we let $\beta=100$.
Suppose that we configure $\Theta\!=\!10^5$, $L\!=\!10^3$, $\Gamma\!=\!10^3$,
AF-Stream loses at most $201,000$ (i.e., $\Theta+L\beta+\Gamma$) sampled items
and increases the error from $\frac{z}{10,000}$ to $\frac{z}{9,989}$ based on
Theorem~\ref{theorem:join}.  

\section{Proof of Theorem \ref{theorem:sgd}}
\label{sec:appendix_convergence}

We refer readers to the notation and problem setting defined in
\S\ref{sec:theory}. 
Our proof makes the same assumptions as in prior variants of SGD
\cite{Dai2015,Ho2013,Langford2009}, including: (i) the function $f_t$ is 
{\em convex} for all $t$; (ii) the sub-differentials of $f_i$ are bounded:
$\|\nabla f_t\| \le H$ for some $H>0$ ({\em Lipschitz continuous}); (iii) the
distance function $D$ is defined as $D(\PS\|\PS')=\tfrac{1}{2}(\PS-\PS')^2$,
and the maximum distance of model parameters is bounded: 
$\max_{\PS, \PS'}D(\PS\|\PS') \le F^2$ for some $F$ ({\em bounded diameter}).

The goal of our proof is to compute the difference between the average risk of
the actual model parameters and that of the global optimum $\PS^*$, i.e., 
\begin{equation}
\AvgT f_t(\NS_t) - \AvgT f_t(\PS^*), 
\nonumber
\end{equation}
and show that the difference converges to zero when $T$ approaches infinity.
We first show that a non-failed sequence has bounded errors via
Lemmas~\ref{lemma:sgd_one_item} and~\ref{lemma:sgd_seq}.  Since the function
$f_t$ is convex for all $t$, we can bound the difference by:
\begin{equation}
\AvgT f_t(\NS_t)-\AvgT f_t(\PS^*) \le \AvgT\langle\Ng_t,\NS_t-\PS^*\rangle,
\nonumber 
\end{equation} 
where $\Ng_t$ denotes the gradients computed using the actual parameters
(i.e., $\Ng_t=\nabla f_t(\NS_t)$), and the symbol $\langle,\rangle$ denotes
the dot product of two vectors.

First, we decompose $\langle\Ng_t,\NS_t-\PS^*\rangle$ as follows:
\begin{lemma}[\cite{Zinkevich2003}]
	\label{lemma:sgd_one_item}
	If $\NS_{t+1} = \NS_t-\eta_t\Ng_t$, then
	\begin{equation}
	\langle\Ng_t,\NS_t-\PS^*\rangle =
	\tfrac{\eta_t}{2}\|\Ng_t\|^2 + 
	\tfrac{D(\PS^*\|\NS_t)-D(\PS^*\|\NS_{t+1})}{\eta_t}.
	\nonumber
	\end{equation}
\end{lemma}

The lemma is part of Theorem~1 in \cite{Zinkevich2003}.  Prior studies 
(e.g., \cite{Dai2015,Ho2013,Langford2009}) also have similar lemmas.  Next, we prove
that any non-failed sequence of states has bounded errors:
\begin{lemma}
	\label{lemma:sgd_seq}
	Consider $T$ sequential states $\NS_{a+1}, \NS_{a+2}, \cdots,$ $\NS_{a+T}$, 
	where $\NS_{a+t+1} = \NS_{a+t}-\eta_{a+t}\Ng_{a+t}$ for any $0<t<T$,
	\begin{equation}
	\SumT\langle\Ng_{a+t},\NS_{a+t}-\PS^*\rangle\le
	H^2(\sqrt{a+T}-\sqrt{a})+F^2\sqrt{a+T}.
	\nonumber
	\end{equation}
\end{lemma}

\paragraph{Proof of Lemma~\ref{lemma:sgd_seq}.}
By Lemma~\ref{lemma:sgd_one_item}, we show that 
\begin{equation}
\begin{split}
&   \SumT\langle\Ng_{a+t},\NS_{a+t}-\PS^*\rangle 
\\  =  &     \SumT[\tfrac{\eta_{a+t}}{2}\|\Ng_{a+t}\|^2 +
\tfrac{D(\PS^*\|\NS_{a+t})-D(\PS^*\|\NS_{a+t+1})}{\eta_{a+t}}] 
\\  =  &   \SumT\tfrac{\eta_{a+t}}{2}\|\Ng_{a+t}\|^2 +
\tfrac{D(\PS^*\|\NS_{a+1})}{\eta_{a+1}}-\tfrac{D(\PS^*\|\NS_{a+T+1})}{\eta_{a+T}}
\\ & 
+  \sum_{t=2}^{T}[D(\PS^*\|\NS_{a+t})(\tfrac{1}{\eta_{a+t}}-\tfrac{1}{\eta_{a+t-1}})].
\nonumber
\end{split}
\end{equation}
We bound each term as follows. For the first term,
\begin{equation}
\begin{split}
&  \SumT\tfrac{\eta_{a+t}}{2}\|\Ng_{a+t}\|^2 
\\ &  \le  \SumT\tfrac{\eta_{a+t}}{2}H^2  = \SumT\tfrac{1}{2\sqrt{a+t}}H^2
\le H^2(\sqrt{a+T}-\sqrt{a}).
\nonumber
\end{split}
\end{equation}
For the second term, 
\begin{equation}
\begin{split}
& \tfrac{D(\PS^*\|\NS_{a+1})}{\eta_{a+1}}-\tfrac{D(\PS^*\|\NS_{a+T+1})}{\eta_{a+T}} 
\\ & \le \tfrac{D(\PS^*\|\NS_{a+1})}{\eta_{a+1}} 
\le  \tfrac{F^2}{\eta_{a+1}} = F^2\sqrt{a+1}. 
\nonumber
\end{split}
\end{equation}
For the final term,
\begin{equation}
\begin{split}
& \sum_{t=2}^{T}[D(\PS^*\|\NS_{a+t})(\tfrac{1}{\eta_{a+t}}-\tfrac{1}{\eta_{a+t-1}})]
\\ & \le  F^2\!\sum_{t=2}^{T}[\tfrac{1}{\eta_{a+t}}-\tfrac{1}{\eta_{a+t-1}}] 
\le \ F^2(\sqrt{a+T}-\sqrt{a+1}).
\nonumber
\end{split}
\end{equation}
Summing up all terms, we have
\begin{equation}
\begin{split}
&    \SumT\langle\Ng_{a+t},\NS_{a+t}-\PS^*\rangle
\\ & \le  \  H^2(\sqrt{a+T}-\sqrt{a})+F^2\sqrt{a+1}
\\ & \hspace{0.5in} +F^2(\sqrt{a+T}-\sqrt{a+1})
\\ & =  \ H^2(\sqrt{a+T}-\sqrt{a})+F^2\sqrt{a+T}.  
\nonumber
\end{split}
\end{equation}

\begin{figure}[t]
	\centering
	\includegraphics[width=1\linewidth]{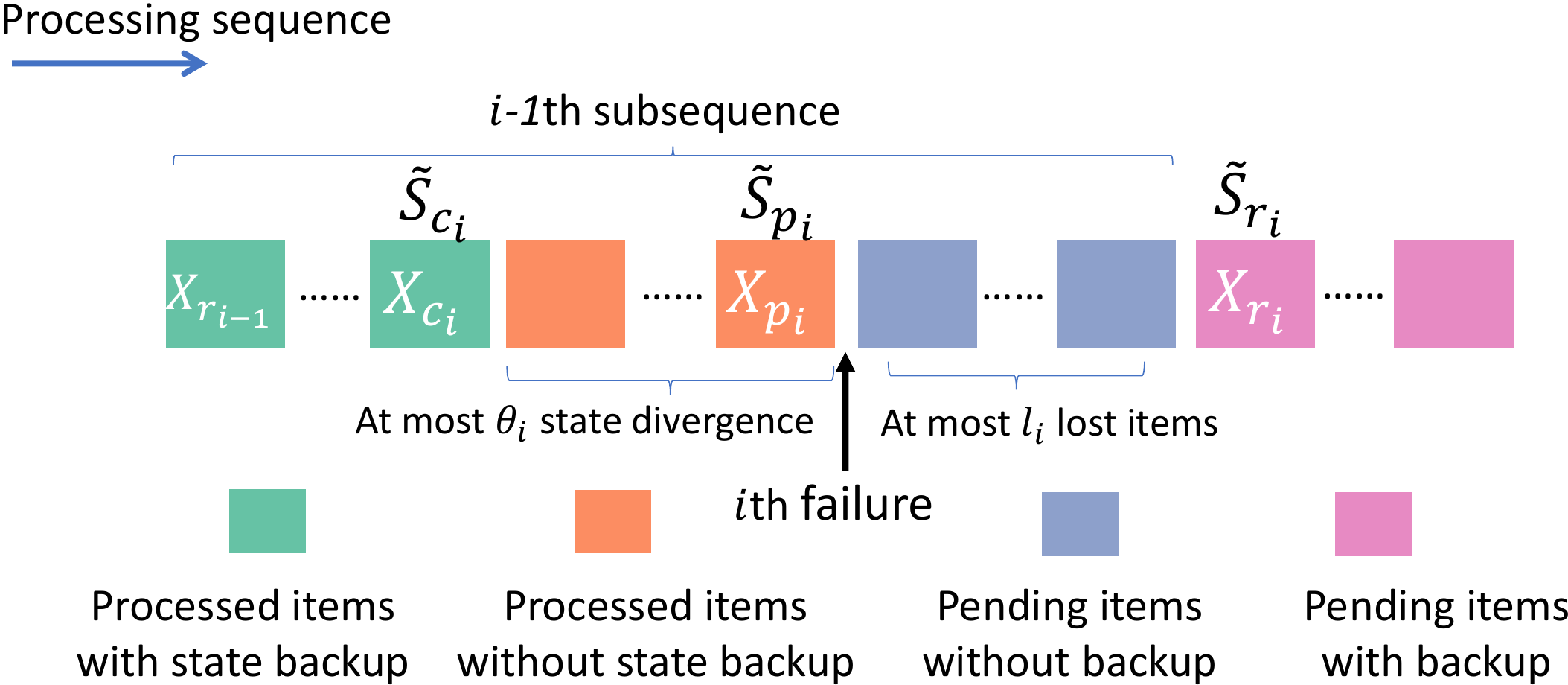}
	\vspace{-6pt}
	\caption{Example when the $i$-th failure happens (the items are processed from
		left to right).}
	\label{fig:state}
\end{figure}

\paragraph{Proof of Theorem~\ref{theorem:sgd}.} We prove
Theorem~\ref{theorem:sgd} based on above lemmas.  We define the following
notation (see Figure~\ref{fig:state} for illustration).  Suppose that the
$i$-th failure happens right after an item, denoted by $X_{p_i}$ is processed.
To recover from the failure, AF-Stream loads the most recent checkpoint state
(i.e., the model parameters), denoted by $\NS_{c_i}$, which corresponds to
the state vector of SGD after the item $X_{c_i}$ has been processed.  It also
resumes SGD from some recovered item, denoted by $X_{r_i}$, where $c_i < r_i$.
We can equivalently view that the model updates from $X_{c_i+1}$ to
$X_{r_i-1}$ are all discarded, and hence $\NS_{t}=\NS_{c_i}$ for $c_i+1\le t
\le r_i$.  Without loss of generality, we set $r_0=1$.

Upon the recovery of the $i$-th failure, let $\theta_i$ denote the divergence
between the recovered state and the actual state before the failure, $l_i$ be
the number of lost items that are not yet processed and have no backup, and
$\alpha$ be the state divergence incurred by each lost item.  Note that
$\theta_i$ and $l_i$ are dynamic thresholds with respect to $\Theta$ and $L$,
respectively, and they both decrease as $i$
increases to bound the overall errors (see details below).  Also, $\alpha$ is
expected to be bounded in any well-defined machine learning model; otherwise
the model becomes unstable.  Thus, the total state divergence is $\theta_i +
l_i\alpha$. 

Our proof has three parts: (i) split all $T$ items into $k+1$ subsequences due
to $k$ failures, (ii) apply Lemma~\ref{lemma:sgd_seq} to bound the error of
each subsequence, and (iii) sum up the errors of all subsequences.

(i) We first split the sequence $\{\NS_t: 1\le t\le T\}$ into $k+1$
subsequences:
\begin{equation}
\begin{split}
& \SumT f_t(\NS_t) - \SumT f_t(\PS^*) 
\\ &   =    \sum_{i=0}^{k-1}\sum_{t=r_i}^{r_{i+1}-1}[f_t(\NS_t)-f_t(\PS^*)] +
\sum_{t=r_{k}}^{T}[f_t(\NS_t)-f_t(\PS^*)].
\nonumber
\end{split}
\end{equation}

(ii) We next bound the difference of each subsequence.
However, for each $\sum_{t=r_i}^{r_{i+1}-1}[f_t(\NS_t)-f_t(\PS^*)]$, we cannot
apply Lemma~\ref{lemma:sgd_seq} directly since the condition $\NS_{t+1} =
\NS_{t}-\eta_{t}\Ng_{t}$ does not hold for $c_{i+1}<t<r_{i+1}$.  To
compensate, we construct another subsequence $\NS'_t$ for $c_{i+1}<t<r_{i+1}$,
where $\NS'_{c_{i+1}+1} = \NS_{c_{i+1}}-\eta_{c_{i+1}}\Ng_{c_{i+1}}$ and
$\NS'_{t+1} = \NS'_{t}-\eta_{t}\Ng'_{t}$ for $c_{i+1}+1<t<r_{i+1}-1$.
Then Lemma~\ref{lemma:sgd_seq} works for the concatenation of the two
subsequences $\{\NS_t: r_i\le t \le c_{i+1}\}$ and 
$\{\NS'_t: c_{i+1}<t<r_{i+1}\}$. 

{\allowdisplaybreaks
	\begin{align*}
	& \sum_{t=r_i}^{r_{i+1}-1}[f_t(\NS_t)-f_t(\PS^*)] 
	\\  =  & \sum_{t=r_i}^{c_{i+1}}[f_t(\NS_t)-f_t(\PS^*)]
	+\sum_{t=c_{i+1}+1}^{r_{i+1}-1}[f_t(\NS'_t)-f_t(\PS^*)] 
	\\ &  +\sum_{t=c_{i+1}+1}^{r_{i+1}-1}[f_t(\NS_t)-f_t(\PS^*)]
	- \sum_{t=c_{i+1}+1}^{r_{i+1}-1}[f_t(\NS'_t)-f_t(\PS^*)] \\  
	\le & \ \ H^2(\sqrt{r_{i+1}-1}-\sqrt{r_i-1}) + F^2\sqrt{r_{i+1}-1} \\ & + \sum_{t=c_{i+1}+1}^{r_{i+1}-1}[f_t(\NS_t)-f_t(\NS'_t)] \\
	\le & \ \ H^2(\sqrt{r_{i+1}-1}-\sqrt{r_i-1}) + F^2\sqrt{r_{i+1}-1} \\ & + \sum_{t=c_{i+1}+1}^{r_{i+1}-1}\langle\nabla f_t(\NS_t),\NS_t-\NS'_t\rangle \\
	\le & \ \ H^2(\sqrt{r_{i+1}-1}-\sqrt{r_i-1}) + F^2\sqrt{r_{i+1}-1} \\ & + \sum_{t=c_{i+1}+1}^{r_{i+1}-1}\|\nabla f_t(\NS_t)\|\cdot D(\NS_t\|\NS'_t)
	\\
	\le & \ \ H^2(\sqrt{r_{i+1}-1}-\sqrt{r_i-1}) + F^2\sqrt{r_{i+1}-1}  \\ &  + H\sum_{t=c_{i+1}+1}^{r_{i+1}-1}D(\NS_{c_i}\|\NS'_t) \\
	\le & \ \ H^2(\sqrt{r_{i+1}-1}-\sqrt{r_i-1}) + F^2\sqrt{r_{i+1}-1} 
	\\ &  + \mathcal{O}(H(\theta^2_{i+1}+\alpha l^2_{i+1})).
	\nonumber
	\end{align*}
}

(iii) Finally, we sum up the differences of all subsequences.
\begin{equation}
\begin{split}
& \SumT f_t(\NS_t) - \SumT f_t(\PS^*) 
\\  =  &   \sum_{i=0}^{k-1}\sum_{t=r_i}^{r_{i+1}-1}[f_t(\NS_t)-f_t(\PS^*)] +
\sum_{t=r_{k}}^{T}(f_t(\NS_t)-f_t(\PS^*)) \\
\le & \sum_{i=0}^{k-1}[H^2(\sqrt{r_{i+1}-1}-\sqrt{r_i-1}) + F^2\sqrt{r_{i+1}-1} \\ 
& + \mathcal{O}(H(\theta^2_{i+1} +\alpha l^2_{i+1}))]   +
H^2(\sqrt{T}-\sqrt{r_k-1}) + F^2\sqrt{T} \\
\le & \ \mathcal{O}(H(\sum_{i=1}^{k}\theta^2_{i+1}+\alpha\sum_{i=1}^{k}l^2_{i+1})) + F^2\sum_{i=1}^{k}\sqrt{r_{i}} \\
&	+ (H^2+F^2)\sqrt{T}.
\nonumber
\end{split}
\end{equation}

To prove Theorem~\ref{theorem:sgd}, we need to bound three terms:
$\sum_{i=1}^{k}\theta^2_{i+1}$, $\sum_{i=1}^{k}l^2_{i+1}$, and
$\sum_{i=1}^{k}\sqrt{r_{i}}$. The first two terms can be bounded by halving
$\theta_{i+1}$ and $l_{i+1}$ after each failure recovery to make the sum of a
geometric series bounded (\S\ref{subsubsec:multiple} and
\S\ref{subsubsec:multiple2}).  Specifically, if the initial values are set to
be $\theta_{1}=\mathcal{O}(\Theta)$ and $l_{1}=\mathcal{O}(L)$, we have
$\mathcal{O}(H(\sum_{i=1}^{k}\theta^2_{i+1}+\alpha\sum_{i=1}^{k}l^2_{i+1}))
=\mathcal{O}(H(\Theta^2+\alpha L^2))$.  
For the term $\sum_{i=1}^{k}\sqrt{r_{i}}$, since each $r_i$ must be smaller
than $T$, we have $\sum_{i=1}^{k}\sqrt{r_{i}} \le k\sqrt{T}$.

We can now bound the sum of the differences of all subsequences as follows:
\begin{equation}
\begin{split}
& \SumT f_t(\NS_t) - \SumT f_t(\PS^*) \\
\le & \ \mathcal{O}(H(\Theta^2+\alpha L^2)) + F^2 k\sqrt{T} + (H^2+F^2)\sqrt{T}.
\nonumber
\end{split}
\end{equation}

By dividing both sides by $T$, the error is negligible when $T$ approaches
infinity, i.e.,
\begin{equation*}
\lim\limits_{T\to\infty}(\AvgT f_t(\NS_t)-\AvgT f_t(\PS^*))=0.  
\end{equation*}

\paragraph{Remarks.} In the proof of Theorem~\ref{theorem:sgd}, the parameters
$\Theta$ and $L$ are used to bound the term $D(\NS_{c_i}\|\NS'_t)$. Note that
the effect of $\Gamma$ has already been included in state divergence (i.e.,
the distance $\Theta$ is the accumulated result of $\Gamma$ unacknowledged
processed items), so our proof does not include $\Gamma$.

\paragraph{Convergence rate.} AF-Stream has a (slightly) slower convergence
rate than the standard SGD.  Specifically, the regret difference to the
optimum is
$\mathcal{O}(H(\Theta^2+\alpha L^2)) + F^2 k\sqrt{T} +(H^2+F^2)\sqrt{T}$,
while the regret difference is $(H^2+F^2)\sqrt{T}$ in the failure-free case
based on Lemma~\ref{lemma:sgd_seq}.  This implies that AF-Stream requires more
steps to achieve the same error level. Nevertheless, our experiments show that
the number of extra steps is limited.

\paragraph{Extensions for other optimization methods.} While our analysis
focuses on SGD, another well-known class of methods for learning parameters in
machine learning is Markov Chain Monte Carlo (MCMC).  In particular, Gibbs
sampling is one well-known MCMC method, and its efficiency for the Latent
Dirichlet Allocation (LDA) model has been well-justified
\cite{Canini2009,Yao2009,Smola2010,Ahmed2012}.  Even though we cannot apply our analysis
of SGD directly to MCMC methods, we briefly discuss how their models converge
under AF-Stream.  In particular, the divergence from the states to a
stationary distribution only depends on how many times the model is updated,
i.e., the number of successfully processed items \cite{Canini2009}.  Since the
errors in AF-Stream are equivalent to dropping a bounded number of model
updates, AF-Stream still ensures model convergence with sufficient input items.

\onecolumn

\section{Interfaces and Built-in Primitive Operators in AF-Stream}
\label{sec:appendix_interfaces}
\begin{table*}[!h]
	\centering
	\renewcommand{\arraystretch}{1.15}
	{\footnotesize
		\begin{tabular}[c]{|l|l|p{2.6in}|}
			\hline
			\textbf{Entities}  &  \textbf{Functions} & \textbf{Descriptions} \\
			\hline
			\hline
			Worker & \texttt{\textbf{void} AddUpstreamWorker(\textbf{string\&} upName)} & Adds an upstream worker\\ 
			\cline{2-3}
			& \texttt{\textbf{void} AddDownstreamWorker(\textbf{string\&} downName)} &  Adds a downstream worker\\
			\cline{2-3}
			& \texttt{\textbf{void} AddUpstreamThread(\textbf{Thread\&} thread)} & Plugs in the upstream thread\\
			\cline{2-3}
			& \texttt{\textbf{void} AddDownstreamThread(\textbf{Thread\&} thread)} & Plugs in the downstream thread\\
			\cline{2-3}
			& \texttt{\textbf{void} AddComputeThread(\textbf{Thread\&} thread)} & Plugs in a compute thread \\
			\cline{2-3}
			& \texttt{\textbf{void} PinCPU(\textbf{Thread\&} thread, \textbf{int} core)} & Pins a thread to a CPU core\\
			\cline{2-3}
			& \texttt{\textbf{void} SetWindow(\textbf{int} type, \textbf{int} length)} & Sets the type and length of
			a window (the window type can be the hopping window, the sliding window,
			or the decaying window) \\
			\cline{2-3}
			& \texttt{\textbf{void} Start()} & Starts the execution of the worker\\
			\hline
			Compute & \texttt{\textbf{void} ConnectFromUpstreamThread()} & Associates the compute thread with the upstream thread\\
			\cline{2-3}
			thread & \texttt{\textbf{void} ConnectToDownstreamThread()} & Associates the compute thread with the downstream thread\\
			\cline{2-3}
			& \texttt{\textbf{void} ConnectToComputeThread(\textbf{Thread\&} dstThread)} &  Connects the compute thread to another compute thread\\
			\cline{2-3}
			& \texttt{\textbf{void} SendToUp(\textbf{string\&} upName, \textbf{Item\&} feedback)} & Sends a feedback item to an upstream worker \\
			\cline{2-3}
			& \texttt{\textbf{void} SendToDown(\textbf{string\&} downName, \textbf{Item\&} data)} & Sends a data item to a downstream worker\\
			\hline
		\end{tabular}
	}
	\caption{Composing Interfaces in C++ Syntax.}
	\label{tab:compose}
\end{table*}

\begin{table*}[!h]
	\centering
	\renewcommand{\arraystretch}{1.15}
	{\footnotesize
		\begin{tabular}[c]{|l|l|p{2.6in}|}
			\hline
			\textbf{Entities }  &  \textbf{Functions} & \textbf{Descriptions}\\
			\hline
			\hline
			Upstream thread & \texttt{\textbf{Item} ReceiveDataItem()} & Receives an input item from data sources or upstream workers \\
			\cline{2-3}
			& \texttt{\textbf{int} GetDestComputeThread(\textbf{Item\&} item)} & Returns the compute thread which an item is dispatched \\
			\cline{2-3}
			& \texttt{\textbf{void} SendFeedbackItem(\textbf{Item\&} feedback)} & Sends a feedback item to an upstream worker \\
			\hline
			Downstream thread & \texttt{\textbf{void} SendDataItem(\textbf{Item\&} feedback)} & Sends output items \\
			\cline{2-3}
			& \texttt{\textbf{Item} ReceiveFeedbackItem()} & Receives a feedback item from downstream workers \\
			\hline
			Compute thread & \texttt{\textbf{bool} ProcessData(\textbf{Item\&} data)} & Processes a data item \\
			\cline{2-3}
			& \texttt{\textbf{bool} ProcessFeedback(\textbf{Item\&} feedback)} & Processes a feedback item \\
			\cline{2-3}
			& \texttt{\textbf{bool} ProcessPunctuation(\textbf{Item\&} punc)} &  Processes a punctuation item\\
			\hline
			Fault-tolerant operator & \texttt{\textbf{double} StateDivergence()} & Gets the divergence of the
			up-to-date state and the backup state \\
			\cline{2-3}
			& \texttt{\textbf{State} BackupState()} & Returns the state to be saved by AF-Stream\\ 
			\cline{2-3}
			& \texttt{\textbf{void} RecoverState(\textbf{State\&} state)} & Obtains the most recent backup state from AF-Stream\\
			\hline
		\end{tabular}
	}
	\caption{User-defined Interfaces in C++ Syntax.}
	\label{tab:interface}
\end{table*}

\begin{table*}[!h]
	\centering
	\renewcommand{\arraystretch}{1.15}
	{\footnotesize
		\begin{tabular}[c]{|l|p{0.36\textwidth}|p{0.20\textwidth}|p{0.22\textwidth}|}
			\hline
			\textbf{Operator}  & \texttt{StateDivergence} & \texttt{BackupState} & \texttt{RecoverState} \\
			\hline
			\hline
			Numeric variables  & Difference of two values & Returns the variable value & Assigns the variable value\\
			\hline
			Vector and matrix & Manhattan distance; Euclidean distance; or maximum difference of values of an index
			& Returns a list of (index, value) pairs and the length of the list
			& Fills in restored (index, value) pairs\\
			\hline
			Hash table &  Manhattan distance, Euclidean distance, maximum difference of
			values of a key, difference of the numbers of keys
			& Returns a list of (key, value) pairs and the length of the list
			& Inserts the restored (key, value) pairs\\
			\hline
			Set & Difference of the numbers of keys
			& Returns a list of set members and the length of the list
			& Inserts the restored set members\\
			\hline
		\end{tabular}
	}
	\caption{Built-in fault-tolerant primitive operators.}
	\label{tab:builtin}
\end{table*}

\end{appendices}

\end{document}